\documentclass[runningheads,final,UKenglish,envcountsame,envcountsect]{llncs}
\pdfoutput=1
\usepackage[final]{hyperref}

\makeatletter

\renewcommand\section{\@startsection{section}{1}{\z@}%
                       {-10\p@ \@plus -5\p@ \@minus -4\p@}%
                       {4\p@ \@plus 3\p@ \@minus 2\p@}%
                       {\normalfont\large\bfseries\boldmath
                        \rightskip=\z@ \@plus 8em\pretolerance=10000 }}
\renewcommand\subsection{\@startsection{subsection}{2}{\z@}%
                       {-8\p@ \@plus -2\p@ \@minus -2\p@}%
                       {7\p@ \@plus 2\p@ \@minus 2\p@}%
                       {\normalfont\normalsize\bfseries\boldmath
                        \rightskip=\z@ \@plus 8em\pretolerance=10000 }}

\renewcommand\subsubsection{\@startsection{subsubsection}{3}{\z@}%
                       {-4\p@ \@plus -4\p@ \@minus -2\p@}%
                       {-0.5em \@plus -0.22em \@minus -0.1em}%
                       {\normalfont\normalsize\bfseries\boldmath}}

\makeatother

\usepackage[T1]{fontenc}

\usepackage[utf8]{inputenc}

\usepackage{microtype}
\usepackage{lmodern}
\usepackage{xcolor}
\usepackage{rotating}
\usepackage{ifdraft}
\usepackage{subcaption}
\usepackage[noend]{algpseudocode}
\usepackage{mathtools}
\usepackage{float}
\usepackage{wrapfig}
\usepackage{nicefrac}
\usepackage{tabularx}
\usepackage{booktabs}
\usepackage{multirow}
\usepackage{makecell}
\usepackage{bm}
\usepackage{etoolbox}
\usepackage{dsfont}

\usepackage[marginparwidth=2cm,nohead,nofoot,
            headsep = 1cm,
            footskip = 1cm,
            text={12.2cm,19.3cm},
            paperwidth=16.2cm,
            paperheight=23.3cm,
            marginparsep=1mm,
            marginparwidth=1.8cm,
            footskip=1cm,
            centering
            ]{geometry}

\usepackage{amssymb,amsmath}

\usepackage{xspace,relsize,stmaryrd}
\usepackage{mathpartir}
\usepackage{multirow}

\usepackage{tikz}
\usetikzlibrary{cd, automata, fit, decorations.pathreplacing}

\makeatletter
\pgfdeclareshape{unaryfunctor}{
  \inheritsavedanchors[from=rectangle] 
  \inheritanchorborder[from=rectangle]
  \inheritanchor[from=rectangle]{center}
  \inheritanchor[from=rectangle]{north}
  \inheritanchor[from=rectangle]{south}
  \inheritanchor[from=rectangle]{west}
  \inheritanchor[from=rectangle]{east}

  \savedanchor\centerpoint{\pgfpointorigin}

  \anchor{in}{%
    \southwest \pgf@xa=\pgf@x \pgf@ya=\pgf@y
    \northeast \pgf@xb=\pgf@x \pgf@yb=\pgf@y
    \advance\pgf@x by0pt
    \pgf@y=\dimexpr 0.5\pgf@ya + 0.5\pgf@yb
  }
  \anchor{out}{%
    \southwest \pgf@xa=\pgf@x \pgf@ya=\pgf@y
    \northeast \pgf@xb=\pgf@x \pgf@yb=\pgf@y
    \pgf@x=\dimexpr \pgf@xa
    \pgf@y=\dimexpr 0.5\pgf@ya + 0.5\pgf@yb
  }

  \backgroundpath{
    \southwest \pgf@xa=\pgf@x \pgf@ya=\pgf@y%
    \northeast \pgf@xb=\pgf@x \pgf@yb=\pgf@y%

    \pgfpathmoveto{\pgfpoint{\pgf@xa}{\pgf@ya}}%
    \pgfpathlineto{\pgfpoint{\pgf@xa}{\pgf@yb}}%
    \pgfpathlineto{\pgfpoint{\pgf@xb}{\pgf@yb}}%
    \pgfpathlineto{\pgfpoint{\pgf@xb}{\pgf@ya}}%
    \pgfpathclose%
  }
}
\makeatother

\makeatletter
\pgfdeclareshape{binaryfunctor}{
  \inheritsavedanchors[from=rectangle] 
  \inheritanchorborder[from=rectangle]
  \inheritanchor[from=rectangle]{center}
  \inheritanchor[from=rectangle]{north}
  \inheritanchor[from=rectangle]{south}
  \inheritanchor[from=rectangle]{east}

  \savedanchor\centerpoint{\pgfpointorigin}

  \anchor{in1}{%
    \southwest \pgf@xa=\pgf@x \pgf@ya=\pgf@y
    \northeast \pgf@xb=\pgf@x \pgf@yb=\pgf@y
    \pgf@xc=\dimexpr 0.5773502 \pgf@xb - 0.5773502 \pgf@xa + 0.5 \pgf@yb - 0.5 \pgf@ya
    \pgf@yc=\dimexpr 0.5 \pgf@ya +  0.5 \pgf@yb
    \pgf@y=\dimexpr 0.75\pgf@ya + 0.25\pgf@yb
    \pgf@x=\dimexpr \pgf@xb
    \pgf@y=\dimexpr \pgf@yc + 0.5\pgf@xc
  }
  \anchor{in2}{%
    \southwest \pgf@xa=\pgf@x \pgf@ya=\pgf@y
    \northeast \pgf@xb=\pgf@x \pgf@yb=\pgf@y
    \pgf@xc=\dimexpr 0.5773502 \pgf@xb - 0.5773502 \pgf@xa + 0.5 \pgf@yb - 0.5 \pgf@ya
    \pgf@yc=\dimexpr 0.5 \pgf@ya +  0.5 \pgf@yb
    \pgf@x=\dimexpr \pgf@xb
    \pgf@y=\dimexpr \pgf@yc - 0.5\pgf@xc
  }
  \anchor{out}{%
    \southwest \pgf@xa=\pgf@x \pgf@ya=\pgf@y
    \northeast \pgf@xb=\pgf@x \pgf@yb=\pgf@y
    \pgf@xc=\dimexpr 0.5773502 \pgf@xb - 0.5773502 \pgf@xa + 0.5 \pgf@yb - 0.5 \pgf@ya
    \pgf@yc=\dimexpr 0.5 \pgf@ya +  0.5 \pgf@yb
    \pgf@x=\dimexpr \pgf@xb - 1.73205\pgf@xc
    \pgf@y=\dimexpr \pgf@yc
  }

  \backgroundpath{
    \southwest \pgf@xa=\pgf@x \pgf@ya=\pgf@y%
    \northeast \pgf@xb=\pgf@x \pgf@yb=\pgf@y%
    \pgf@xc=\dimexpr 0.5773502 \pgf@xb - 0.5773502 \pgf@xa + 0.5 \pgf@yb - 0.5 \pgf@ya
    \pgf@yc=\dimexpr 0.5 \pgf@ya +  0.5 \pgf@yb

    \pgfpathmoveto{\pgfpointadd{\pgfpoint{\pgf@xb}{\pgf@yc}}{\pgfpoint{-1.73205\pgf@xc}{0}}}%
    \pgfpathlineto{\pgfpointadd{\pgfpoint{\pgf@xb}{\pgf@yc}}{\pgfpoint{0}{\pgf@xc}}}%
    \pgfpathlineto{\pgfpointadd{\pgfpoint{\pgf@xb}{\pgf@yc}}{\pgfpoint{0}{-\pgf@xc}}}%
    \pgfpathclose%
  }
}
\makeatother
\tikzset{
      functornode/.style={
        draw=black!80,
        line width=0.5pt,
        anchor=out,
        xshift=8mm,
        fill=none,
        execute at begin node=$,%
        execute at end node=$,%
        rounded corners= 2pt,
      },
      unary/.style={
        functornode,
        shape=unaryfunctor,
        minimum height = 2em,
        minimum width = 2.2em,
      },
      binary/.style={
        functornode,
        shape=binaryfunctor,
        inner ysep=1pt,
      },
}

\usepackage[author=anonymous,marginclue,footnote]{fixme}
\FXRegisterAuthor{ls}{als}{LS}
\FXRegisterAuthor{sm}{asm}{SM}
\FXRegisterAuthor{tw}{atw}{TW}
\FXRegisterAuthor{hp}{ahp}{HP}
\DeclareUnicodeCharacter{D7}{\ensuremath{\times}}

\usepackage{seqsplit}
\usepackage{xstring}

\makeatletter
\usepackage[notcite,notref]{showkeys}%

\newcommand\noshowkeys{\def\hideNextShowKeysLabel{test}}
\renewcommand*\showkeyslabelformat[1]{%
\@ifundefined{hideNextShowKeysLabel}{%
\noexpandarg%
\StrSubstitute{#1}{ }{\textvisiblespace}[\TEMP]%
\parbox[t]{\marginparwidth}{\raggedright\normalfont\small\ttfamily\(\{\){\color{red!50!black}\expandafter\seqsplit\expandafter{\TEMP}}\(\}\)}%
}{}
}

\makeatother

\newcommand{\copar}{\textsf{CoPaR}\xspace}

\hypersetup{final,hidelinks}
\hypersetup{
  bookmarks=true,
  colorlinks=true,
  linkcolor=black,
  citecolor=black,
  filecolor=black,
  urlcolor=black,
  pdftitle={Generic Partition Refinement and Weighted Tree Automata},
  pdfauthor={Hans-Peter Deifel, Stefan Milius, Lutz Schröder, Thorsten
    Wißmann},
  pdfduplex={DuplexFlipLongEdge},
}

\allowdisplaybreaks[2]

\spnewtheorem{thm}[theorem]{Theorem}{\bfseries}{\itshape}
\spnewtheorem{cor}[theorem]{Corollary}{\bfseries}{\itshape}
\spnewtheorem{cnj}[theorem]{Conjecture}{\bfseries}{\itshape}
\spnewtheorem{lem}[theorem]{Lemma}{\bfseries}{\itshape}
\spnewtheorem{lemdefn}[theorem]{Lemma and Definition}{\bfseries}{\itshape}
\spnewtheorem{prop}[theorem]{Proposition}{\bfseries}{\itshape}
\spnewtheorem{defn}[theorem]{Definition}{\bfseries}{\upshape}
\spnewtheorem{rem}[theorem]{Remark}{\bfseries}{\upshape}
\spnewtheorem{notation}[theorem]{Notation}{\bfseries}{\upshape}
\spnewtheorem{expl}[theorem]{Example}{\bfseries}{\upshape}
\spnewtheorem{thmdefn}[theorem]{Theorem and Definition}{\bfseries}{\itshape}
\spnewtheorem{propdefn}[theorem]{Proposition and Definition}{\bfseries}{\itshape}
\spnewtheorem{assumption}[theorem]{Assumption}{\bfseries}{\upshape}
\spnewtheorem{algorithm}[theorem]{Algorithm}{\bfseries}{\upshape}

 \renewenvironment{theorem}{\begin{thm}}{\end{thm}}
 \renewenvironment{corollary}{\begin{cor}}{\end{cor}}
 
 \renewenvironment{proposition}{\begin{prop}}{\end{prop}}
 \renewenvironment{definition}{\begin{defn}}{\end{defn}}
 \renewenvironment{remark}{\begin{rem}}{\end{rem}}
 \renewenvironment{example}{\begin{expl}}{\end{expl}}

\newcommand\vartextvisiblespace[1][.5em]{%
  \makebox[#1]{%
    \kern.07em
    \vrule height.3ex
    \hrulefill
    \vrule height.3ex
    \kern.07em
  }
}

\newcounter{blubber}

\usepackage{enumitem}
\setlist[enumerate,1]{label=(\arabic*),ref=(\arabic*),font=\normalfont,align=left,leftmargin=0pt,labelindent=0pt,listparindent=\parindent,labelwidth=0pt,itemindent=!,topsep=3pt,parsep=0pt,itemsep=3pt,start=1}
\setlist[itemize]{labelindent=*,leftmargin=*,topsep=3pt,itemsep=2pt}

\makeatletter
\def\moverlay{\mathpalette\mov@rlay}
\def\mov@rlay#1#2{\leavevmode\vtop{%
   \baselineskip\z@skip \lineskiplimit-\maxdimen
   \ialign{\hfil$\m@th#1##$\hfil\cr#2\crcr}}}
\newcommand{\charfusion}[3][\mathord]{
    #1{\ifx#1\mathop\vphantom{#2}\fi
        \mathpalette\mov@rlay{#2\cr#3}
      }
    \ifx#1\mathop\expandafter\displaylimits\fi}
\makeatother

\newcommand{\id}{\mathsf{id}}
\newcommand{\inl}{\mathsf{inl}}
\newcommand{\inr}{\mathsf{inr}}

\newcommand{\Set}{\ensuremath{\mathsf{Set}}\xspace}

\newcommand{\Pow}{{\mathcal{P}_\omega}}
\newcommand{\Bag}{{\mathcal{B}_\omega}}
\newcommand{\BagM}{\ensuremath{\mathcal{B}(M_{\neq 0})}}
\newcommand{\Powf}{\mathcal{P}_\omega}
\newcommand{\Bagf}{\ensuremath{\mathcal{B}}\xspace}
\newcommand{\Dist}{\ensuremath{{\mathcal{D}_\omega}}\xspace}
\newcommand{\R}{\mathds{R}}
\newcommand{\N}{\mathds{N}}
\newcommand{\Z}{\mathds{Z}}

\newcommand{\CO}{\mathcal{O}}

\newcommand{\op}[1]{\ensuremath{\mathsf{#1}}}

\newcommand{\fpair}[1]{\langle#1\rangle}

\renewcommand{\gets}{\ensuremath{\xspace\mathop{:=}}}
\newcommand{\rifactor}{\ensuremath{p}} 

\tikzset{
  coalgebra drawing/.style={
    state/.append style={
      minimum width=0pt,
      minimum height=0pt,
      inner sep=0.8mm,
    },
    every edge/.append style={
      shorten <= 1pt,
      shorten >= 1pt,
    }
  }
}

\def\epito{\twoheadrightarrow}
\newcommand{\csum}{\raisebox{-1pt}{\text{\large$\mathrm{\Sigma}$}}}
\newcommand{\takeout}[1]{\empty}

\makeatletter
\def\thanks#1{\footnotemark
    \protected@xdef\@thanks{\@thanks
        \protect\footnotetext[\the\c@footnote]{#1}}%
\def\thanks##1{\addtocounter{footnote}{-1}\footnotemark}
}
\makeatother
 
\begin{document}

\title{Generic Partition Refinement and\texorpdfstring{\\}{} Weighted Tree
  Automata}

\author{Hans-Peter Deifel \and
  Stefan Milius\thanks{Supported by the DFG project COAX (MI 717/5-2 and SCHR
    1118/12-2)}\and
  Lutz Schröder\thanks{}\and 
  Thorsten Wißmann\thanks{}
}
\institute{Friedrich-Alexander-Universit\"{a}t
  Erlangen-N\"urnberg, Germany\\
\email{\{hans-peter.deifel,stefan.milius,lutz.schroeder,thorsten.wissmann\}@fau.de}}

\maketitle

\begin{abstract} Partition refinement is a method for minimizing
  automata and transition systems of various types. Recently, we have
  developed a partition refinement algorithm that is generic in the
  transition type of the given system and matches the run time of the
  best known algorithms for many concrete types of systems,
  e.g.~deterministic automata as well as ordinary, weighted, and
  probabilistic (labelled) transition systems. Genericity is achieved
  by modelling transition types as functors on sets, and systems
  as coalgebras. In the present work, we refine the run time analysis
  of our algorithm to cover additional instances, notably weighted
  automata and, more generally, weighted tree automata.  For weights
  in a cancellative monoid we match, and for non-cancellative monoids
  such as (the additive monoid of) the tropical semiring even
  substantially improve, the asymptotic run time of the best known
  algorithms. We have implemented our algorithm in a generic tool that
  is easily instantiated to concrete system types by implementing a
  simple refinement interface. Moreover, the algorithm and the tool
  are modular, and partition refiners for new types of systems are
  obtained easily by composing pre-implemented basic
  functors. Experiments show that even for complex system types, the
  tool is able to handle systems with millions of transitions.
\end{abstract}

\section{Introduction} Minimization is a basic verification task on
state-based systems, concerned with reducing the number of system
states as far as possible while preserving the system behaviour. It is
used for equivalence checking of systems and as a preprocessing step
in further system analysis tasks, such as model checking. 

In general, minimization proceeds in two steps:~(1) remove unreachable
states, and~(2) identify behaviourally equivalent states. Here, we are
concerned with the second step, which depends on which notion of
equivalence is imposed on states; we work with notions of
\emph{bisimilarity} and generalizations thereof. Classically,
bisimilarity for labelled transition systems obeys the principle
``states $x$ and $y$ are bisimilar if for every transition $x \to x'$,
there exists a transition $y \to y'$ with $x'$ and $y'$ bisimilar''
.
It is thus given via a fixpoint
definition, to be understood as a \emph{greatest} fixpoint, and can
therefore be iteratively approximated from above. This is the
principle behind \emph{partition refinement} algorithms: Initially all
states are tentatively considered equivalent, and then this initial
partition is iteratively refined according to observations made on the
states until a fixpoint is reached. Unsurprisingly, such procedures
run in polynomial time.  Its comparative tractability (in contrast,
e.g.~trace equivalence and language equivalence of non-deterministic
systems are PSPACE complete~\cite{KanellakisS90}) makes miminization
under bisimilarity interesting even in cases where the main
equivalence of interest is linear-time, such as word automata.

Kanellakis and Smolka~\cite{KanellakisS90} in fact provide a
minimization algorithm with run time $\CO(m\cdot n)$ for ordinary
transition systems with $n$ states and $m$ transitions. However, even
faster partition refinement algorithms running in
$\CO((m+n)\cdot \log n)$ have been developed for various types of
systems over the past 50 years. For example, Hopcroft's algorithm
minimizes deterministic automata for a fixed input alphabet $A$ in
$\CO(n\cdot \log n)$ \cite{Hopcroft71}; it was later generalized to
variable input alphabets, with run time $\CO(n\cdot |A|\cdot \log n)$
\cite{Gries1973,Knuutila2001}. The Paige-Tarjan algorithm minimizes
transition systems in time $\CO((m+n)\cdot \log n)$
\cite{PaigeTarjan87}, and generalizations to labelled transition
systems have the same time complexity
\cite{HuynhTian92,DerisaviEA03,Valmari09}. Minimization of weighted
systems is typically called \emph{lumping} in the literature, and
Valmari and Franchescini~\cite{ValmariF10} have developed a simple
$\CO((m+n)\cdot \log n)$ lumping algorithm for systems with rational
weights.

In earlier work~\cite{concur2017,concurSpecialIssue} we have developed
an efficient \emph{generic} partition refinement algorithm that can be
easily instantiated to a wide range of system types, most of the time
either matching or improving the previous best run time. The
genericity of the algorithm is based on modelling state-based systems
as coalgebras following the paradigm of universal
coalgebra~\cite{Rutten00}, in which the branching structure of systems
is encapsulated in the choice of a functor, the \emph{type
  functor}. This allows us to cover not only classical relational
systems and various forms of weighted systems, but also to combine
existing system types in various ways, e.g.~nondeterministic and
probabilistic branching. Our algorithm uses a functor-specific
\emph{refinement interface} that supports a graph-based representation
of coalgebras. It allows for a generic complexity analysis, and indeed
the generic algorithm has the same asymptotic complexity as the
above-mentioned specific algorithms;\twnote{nicht ganz war für
  automaten mit variablem $A$ oder LTS } for Segala
systems~\cite{Segala95} (systems that combine probabilistic and
non-deterministic branching, also known as Markov decision processes),
it even improves on the best known run time and matches the run time
of a recent algorithm~\cite{GrooteEA18} discovered independently and
almost at the same time.

The new contributions of the present paper are twofold. On the
theoretical side, we show how to instantiate our generic algorithm to
weighted systems with weights in a monoid (generalizing the
group-weighted case considered
previously~\cite{concur2017,concurSpecialIssue}). We then refine the
complexity analysis of the algorithm, making the complexity of the
implementation of the type functor a parameter $\rifactor(n,m)$, where $n$ and
$m$ are the numbers of nodes and edges, respectively, in the graph
representation of the input coalgebra.  In the new setup, the previous
analysis becomes the special case where $\rifactor(n,m) = 1$. Under the same
structural assumptions on the type functor and the refinement
interface as previously, our algorithm runs in time
$\CO(m \cdot \log n\cdot \rifactor(n,m))$. Instantiated to the case of
weighted systems over non-cancellative monoids (with
$\rifactor(n,m) = \log(n)$), such as the additive monoid $(\N,\max,0)$ of the
tropical semiring, the run time of the generic algorithm is
$\CO(m\cdot\log^2 m)$, thus markedly improving the run time
$\CO(m\cdot n)$ of previous algorithms for weighted
automata~\cite{Buchholz08} and, more generally, (bottom-up) weighted tree
automata~\cite{HoegbergEA07} for this case. In addition, for
cancellative monoids, we again essentially match the complexity of the
previous algorithms~\cite{Buchholz08,HoegbergEA07}.\twnote{This is
  wrong, they have $\CO(m\cdot \log n)$ and we $\CO(m\cdot \log m)$.
  LS: Which is the same if the alphabet is fixed, so I added
  ``essentially''}

Our second main contribution is a generic and modular implementation
of our algorithm, the \emph{Coalgebraic Partition Refiner}
(\copar). Instantiating \copar to coalgebras for a given functor
requires only to implement the refinement interface. We provide such
implementations for a number of basic type functors, e.g.\ for
non-deterministic, weighted, or probabilistic branching, as well as
(ranked) input and output alphabets or output weights. In addition,
\copar is \emph{modular}: For any type functor obtained by composing
basic type functors for which a refinement interface is available,
\copar automatically derives an implementation of the refinement
interface. We explain in detail how this modularity is realized in our
implementation and, extending Valmari and Franchescini's
ideas~\cite{ValmariF10}, we explain how the necessary data structures
need to be implemented so as to realize the low theoretical
complexity. We thus provide a working efficient partition refiner for
all the above mentioned system types. In particular, our tool is, to
the best of our knowledge, the only available implementation of
partition refinement for many composite system types, notably for
weighted (tree) automata over non-cancellative monoids. The tool
including source code and evaluation data is available at
\url{https://git8.cs.fau.de/software/copar}.

\smnote{TODO: mention special tools, e.g.~CADP, mCRL2, LTSmin.}

\takeout{
\paragraph{Organization} In \autoref{sec:theofoundations} we recall
the necessary technical background and the principles behind our
generic
algorithm~\cite{concur2017,concurSpecialIssue}.

\smnote[inline]{TODO: from here need to adjust to new text.}
In \autoref{sec:partref}, the
usage of our tool is described and how it can be extended.
In particular, \autoref{sec:des} discusses in
detail how the mechanism for combining system types is realized within
the tool. Some concrete instantiations are described in
\autoref{sec:instances}.

We illustrate the flexibility of our algorithm and tool in
\autoref{sec:wta} by implementing the refinement interface of the type
functor for weighted (tree) automata which yields an efficient partition refiner
for this type of systems. We demonstrate the efficiency by a number of
benchmarks.\smnote{This par is new text moved here from section 3.}

In \autoref{sec:bench} we evaluate our tool
by benchmarking it for some system types, and in particular provide a
comparison with existing specific tools. Additional proofs can be found in the
full version online~\cite{fullversion}.}

\section{Theoretical Foundations}

\label{sec:theofoundations}
Our algorithmic framework~\cite{concur2017,concurSpecialIssue} is
based on modelling state-based systems abstractly as \emph{coalgebras}
for a (set) \emph{functor} that encapsulates the transition type,
following the paradigm of \emph{universal
  coalgebra}~\cite{Rutten00}. We proceed to recall standard notation
for sets and maps, as well as basic notions and examples in
coalgebra. We fix a singleton set $1=\{*\}$; for every set~$X$ we have
a unique map $!\colon X\to 1$. We denote composition of maps by
$(-)\cdot(-)$, in applicative order. Given maps $f\colon X\to A$,
$g\colon X\to B$ we define $\fpair{f,g}\colon X\to A\times B$ by
$\fpair{f,g}(x) = (f(x),g(x))$.  We model the transition type of state
based systems using \emph{functors}. Informally, a functor~$F$ assigns
to a set~$X$ a set~$FX$ of structured collections over~$X$, and an
$F$-coalgebra is a map~$c$ assigning to each state~$x$ in a system a
structured collection $c(x)\in FX$ of successors. The most basic
example is that of transition systems, where~$F$ is powerset, so a
coalgebra assigns to each state a set of successors. Formal
definitions are as follows.

\begin{definition}
  A \emph{functor} $F\colon \Set\to\Set$ assigns to each
  set~$X$ a set~$FX$, and to each map $f\colon X\to Y$ a map
  $Ff\colon FX\to FY$, preserving identities and composition
  ($F\id_X=\id_{FX}$, $F(g\cdot f)=Fg\cdot Ff$).  An
  \emph{$F$-coalgebra} $(C,c)$ consists of a set~$C$ of
  \emph{states} and a \emph{transition} structure $c\colon C\to FC$.
  A \emph{morphism} $h\colon (C,c)\to (D,d)$ of $F$-coalgebras is a
  map $h\colon C\to D$ that preserves the transition structure, i.e.~$Fh\cdot c
  = d\cdot h$. Two states $x,y\in C$ of a coalgebra $c\colon C\to FC$ are
  \emph{behaviourally equivalent} ($x \sim y$) if there exists a coalgebra
  morphism $h$ with $h(x) = h(y)$.
\end{definition}
\begin{example} \label{exManyFunctors}
  \begin{enumerate}
  \item \label{exManyFunctors:Pow}
    The \emph{finite powerset} functor $\Pow$ 
    maps a set~$X$ to the set~$\Pow X$ of all \emph{finite} subsets
    of~$X$, and a map $f\colon X\to Y$ to the map
    $\Pow f = f[-]\colon \Pow X\to \Pow Y$ taking direct images.
    $\Pow$-coalgebras are finitely branching
    (unlabelled) transition systems. and two states are behaviourally
    equivalent iff they are bisimilar.
  \item For a fixed finite set $A$, the functor given by
    $FX=2\times X^A$, where $2 = \{0,1\}$, sends a set $X$ to the set
    of pairs of boolean values and functions $A\to X$. An
    $F$-coalgebra $(C,c)$ is a deterministic automaton (without an
    initial state). For each state $x\in C$, the first component of
    $c(x)$ determines whether $x$ is a final state, and the second
    component is the successor function $A\to X$ mapping each input
    letter $a\in A$ to the successor state of $x$ under input letter
    $a$. States $x,y \in C$ are behaviourally equivalent iff they
    accept the same language in the usual sense.
  \item\label{exManyFunctors:3} For a commutative monoid $(M,+,0)$,
    the \emph{monoid-valued} functor $M^{(-)}$ sends each set $X$ to
    the set of maps $f\colon X\to M$ that are finitely supported, i.e.~
    $f(x) = 0$ for almost all $x\in X$. An $F$-coalgebra
    $c\colon C\to M^{(C)}$ is, equivalently, a finitely branching
    $M$-weighted transition system: For a state $x\in C$, $c(x)$ maps
    each state $y\in C$ to the weight $c(x)(y)$ of the transition from
    $x$ to $y$. For a map $f\colon X\to Y$,
    $M^{(f)}\colon M^{(X)}\to M^{(Y)}$ sends a finitely supported map
    $v\colon X\to M$ to the map
    $y\mapsto \sum_{x\in X, f(x) = y} v(x)$, corresponding to the
    standard image measure construction. As the notion of behavioural
    equivalence of states in $M^{(-)}$-coalgebras, we obtain weighted
    bisimilarity (cf.~\cite{Buchholz08,KlinS13}), given coinductively
    by postulating that states $x,y\in C$ are behaviourally equivalent
    $(x\sim y)$ iff
      \[\textstyle
        \sum_{z'\sim z} c(x)(z')
        = \sum_{z'\sim z} c(y)(z')\qquad \text{for all $z\in C$}.
      \]
      For the Boolean monoid $(2=\{0,1\},\vee,0)$, the monoid-valued
      functor $2^{(-)}$ is (naturally isomorphic to) the finite
      powerset functor $\Pow$.  For the monoid of real numbers
      $(\R,+,0)$, the monoid-valued functor $\R^{(-)}$ has
      $\R$-weighted systems as coalgebras, e.g.~Markov chains. In
      fact, finite Markov chains are precisely finite coalgebras of
      the \emph{finite distribution functor}, i.e.~the
      subfunctor~$\Dist$ of $\R_{\ge 0}^{(-)}$ (and hence
      of~$\R^{(-)}$) given by
      $\Dist(X) = \{\mu \in \R_{\ge 0}^{(X)}\mid\sum_{x \in X} \mu(x)
      = 1 \}$.
      For the monoid $(\N,+,0)$ of natural numbers, the monoid-valued
      functor is the bag functor $\Bag$, which maps a set $X$ to the
      set of finite multisets over~$X$.
  \end{enumerate}
\end{example}
%
%

\section{Generic Partition Refinement}
\label{sec:partref}

We recall some key aspects of our generic partition refinement
algorithm~\cite{concur2017,concurSpecialIssue}, which \emph{minimizes} a
given coalgebra, i.e.~computes its quotient modulo behavioural
equivalence; we center the presentation around the implementation and
use of our tool.

The algorithm~\cite[Algorithm 4.5]{concurSpecialIssue} is parametrized
over a type functor $F$, represented by implementing a fixed
\emph{refinement interface}, which in particular allows for a
representation of $F$-coalgebras in terms of nodes and edges (by no
means implying a restriction to relational systems!). Our previous
analysis has established that the algorithm minimizes $F$-coalgebras
with~$n$ nodes and~$m$ edges in time $\CO(m\cdot \log n)$, assuming $m\ge n$ and
that the operations of the refinement interface run in
linear time. In the present paper, we generalize the analysis,
establishing a run time in $\CO(m\cdot \log n\cdot \rifactor(n,m))$, where
$\rifactor(n,m)$ is a factor in the time complexity of the operations implementing the
refinement interface, and depends on the functor at hand. In many
instances, $\rifactor(n,m) = 1$, reproducing the previous analysis. In
some cases, $\rifactor(n,m)$ is not constant, and our new
analysis still applies in these cases, either matching or
improving the best known run time in most instances, most notably
weighted systems over non-cancellative monoids.

We proceed to discuss the design of the implementation,
including input formats of our tool \copar for composite functors built from
pre-implemented basic blocks and for systems to be minimized
(\autoref{sec:parsing}). The refinement interface and its
implementation are described in \autoref{sec:refintface}.

\subsection{Generic System Specification}
\label{sec:parsing}
\copar accepts as input a file that represents a finite $F$-coalgebra
$c\colon C\to FC$, and consists of two parts. The first part is a
single line specifying the functor $F$. Each of the remaining lines
describes one state $x\in C$ and its one-step behaviour
$c(x)$. Examples of input files are shown in
\autoref{fig:example-input}.

\begin{figure}[b]%
  \centering%
  \vspace{-2mm}
  \begin{subfigure}[b]{0.5\textwidth}%
    \begin{minipage}[b]{.55\textwidth}%
\begin{verbatim}
DX

q: {p: 0.5, r: 0.5}
p: {q: 0.4, r: 0.6}
r: {r: 1}
\end{verbatim}
    \end{minipage}
    \hfill%
    \begin{tikzpicture}[coalgebra drawing]
      \node (q1) at (135:9mm) [state] {$q$};
      \node (q2) at (-135:9mm) [state] {$p$};
      \node (q3) at (0:6mm) [state] {$r$};
      \path[->,bend angle=20]
      (q1) edge [left, bend right] node[overlay] {$\frac{1}{2}$} (q2)
      (q1) edge [above,bend left] node[overlay] {$\frac{1}{2}$} (q3)
      (q2) edge [right, bend right] node {$\frac{2}{5}$} (q1)
      (q2) edge [below, bend right,overlay] node {$\frac{3}{5}$} (q3)
      (q3) edge [loop above] node {$1$} (q3);
    \end{tikzpicture}%
    \hfill\hspace*{0pt}%
    \caption{Markov chain}
    \label{subfig:markov}
  \end{subfigure}%
  \hfill%
  \begin{subfigure}[b]{0.48\textwidth}
    \begin{minipage}[b]{.58\textwidth}
\begin{verbatim}
{f,n} x X^{a,b}

q: (n, {a: p, b: r})
p: (n, {a: q, b: r})
r: (f, {a: q, b: p})
\end{verbatim}
    \end{minipage}%
    \hfill%
    \begin{tikzpicture}[coalgebra drawing]
      \node (q) at (135:10mm) [state] {$q$};
      \node (p) at (-135:10mm) [state] {$p$};
      \node (r) at (0:8mm) [state,accepting] {$r$};
      \path[->,bend angle=15,
      every node/.append style={
        shape=circle,
        inner sep=1pt,
        fill=white,
        anchor=center,
      }]
      (q) edge [bend right] node {a} (p)
      (q) edge [bend right] node {b} (r)
      (p) edge [bend right] node {a} (q)
      (p) edge [bend left] node {b} (q3)
      (r) edge [bend right] node {a} (q)
      (r) edge [bend left] node {b} (p);
    \end{tikzpicture}
    \hfill%
    \caption{Deterministic finite automaton}
    \noshowkeys\label{fig:dfa-picture}
  \end{subfigure}
  \caption{Examples of input files with encoded coalgebras}
  \vspace{-4mm}
  \label{fig:example-input}
\end{figure}

\subsubsection*{Functor specification.} Functors are specified as
composites of basic building blocks; that is, the functor given in the
first line of an input file is an expression determined by the grammar
\begin{equation}
T ::= \texttt{X} \mid F(T,\ldots,T) \qquad (F\colon \Set^k\to \Set) \in \mathcal{F},
\label{termGrammar}
\end{equation}
where the character $\texttt{X}$ is a terminal symbol and
$\mathcal{F}$ is a set of predefined symbols called \emph{basic
  functors}, representing a number of pre-implemented functors of type
$F\colon \Set^k\to \Set$. Only basic functors need to be implemented
explicitly (\autoref{sec:refintface}); for composite functors, the
tool derives instances of the algorithm automatically
(\autoref{sec:des}). Basic functors currently implemented include
(finite) powerset~$\Pow$, the bag functor~$\Bag$,
monoid-valued functors~$M^{(-)}$, and polynomial functors for finite
many-sorted signatures $\Sigma$, based on the description of the
respective refinement interfaces given in our previous
work~\cite{concur2017,concurSpecialIssue} and, in the case of
$M^{(-)}$ for unrestricted commutative monoids~$M$ (rather than only
groups) the newly developed interface described in
\autoref{sec:wta}. Since behavioural equivalence is preserved and
reflected under converting $G$-coalgebras into $F$-coalgebras for a
subfunctor~$G$ of~$F$~\cite[Proposition~2.13]{concurSpecialIssue}, we
also cover subfunctors, such as the finite distribution
functor~$\Dist$ as a subfunctor of $\R^{(-)}$. With the polynomial
constructs $+$ and $\times$ written in infix notation as usual, the
currently supported grammar is effectively
\begin{align}
  T &::= \texttt{X}
      \mid
      \Pow\, T \mid \Bag\, T \mid \Dist\, T
          \mid M^{(T)} 
      \mid \Sigma \label{termGrammar2}
\\
  \Sigma &::= C \mid T + T \mid T \times T \mid T^A
      \quad~
  C ::= \mathbb{N} ~\vert~ A
      \quad~
  A ::= \{s_1,\ldots,s_n\} \notag
\end{align}
where the $s_k$ are strings subject to the usual conventions for
C-style identifiers, exponents $F^A$ are written \verb|F^A|, and
$M$ is one of the monoids $(\Z,+,0)$, $(\R,+,0)$, $(\mathds{C},+,0)$,
$(\Powf(64), \cup, \emptyset)$ (i.e.\ the monoid of $64$-bit words
with bitwise $\op{or}$),
and $(\N,\max,0)$ (the additive monoid of the tropical semiring). Note
that~$C$ effectively ranges over at most countable sets, and~$A$ over
finite sets. A term~$T$ determines a functor $F\colon \Set\to\Set$ in
the evident way, with~$\texttt{X}$ interpreted as the argument,
i.e.~$F(\texttt{X}) = T$. It should be noted that the implementation
treats composites of polynomial (sub-)terms as a single functor in
order to minimize overhead incurred by excessive decomposition,
e.g.~$X \mapsto \{0,1\} + \Pow(\R^{(X)}) + X \times X$ is composed
from the basic functors $\Pow$, $\R^{(-)}$ and the $3$-sorted
polynomial functor $\Sigma(X,Y,Z) = \{0,1\} + X + Y \times Z$.

\subsubsection*{Coalgebra specification.} The remaining lines of an input
file define a finite $F$-coalgebra $c\colon C\to FC$. Each line of the
form $x\texttt{:}\text{\textvisiblespace} t$ defines a state $x\in C$,
where $x$ is a C-style identifier, and $t$ represents the element
$t=c(x)\in FC$. The syntax for $t$ depends on the specified
functor~$F$, and follows the structure of the term~$T$ defining~$F$;
we write $t\in T$ for a term $t$ describing an element
of~$FC$:
\begin{itemize}
\item $t\in\texttt{X}$ iff~$t$ is one of the named
  states specified in the file.
\item $t\in T_1\times\dots\times T_n$ is given by
  $t ::= (t_1,\dots,t_n)$ where $t_i \in T_i$, $i=1,\dots, n$.
\item $t\in T_1 + \dots + T_n$ is given by
  $t ::= \texttt{inj}\text{\textvisiblespace}
  i\text{\textvisiblespace} t_i$ where $i=1,\dots, n$ and $t_i \in T_i$.
\item $t\in \Pow T $ and $t\in\Bag T$ are given by
  $t ::= \texttt{\{}t_1,\ldots,t_n\texttt{\}}$ with
  $t_1,\dots,t_n \in T$.
\item $t\in M^{(T)}$ is given by
  $t::=\texttt{\{}t_1\texttt{:}\text{\textvisiblespace}m_1\texttt{,}
  \ldots\texttt{,}$
  $t_n\texttt{:}\text{\textvisiblespace}m_n \texttt{\}}$ with
  $m_1,\dots,m_n\in M$ and $t_1,\dots,t_n\in T$, denoting
  $\mu \in M^{(TC)}$ with $\mu(t_i) = m_i$ and $\mu(t)=0$ otherwise.
\end{itemize}
For example, for the functor~$F$ given by the term
$T=\Pow(\{a,b\}\times\R^{(X)}$), the
one-line declaration \texttt{x: \{(a,\{x: 2.4\}), (a,\{\}), (b,\{x:
  -8\})\}} defines an $F$-coalgebra with a single state~$x$, having two
$a$-successors and one~$b$-successor, where successors are elements of
$\R^{(X)}$. One $a$-successor is constantly zero, and the other assigns
weight~$2.4$ to~$x$; the $b$-successor assigns weight $-8$ to~$x$. Two more
examples are shown in Fig.~\ref{fig:example-input}.
 
\subsubsection*{Parsing input files.} After reading the functor term $T$,
the tool builds a parser for the functor-specific input format and
parses an input coalgebra specified in the above syntax into an
intermediate format described in \autoref{sec:refintface}. In the case
of a composite functor, the 
parsed coalgebra\smnote{Is this correct? ``System'' could also mean
  ``our tool'' here.} then undergoes a substantial amount of
preprocessing that also affects how transitions are counted; we defer
the discussion of this point to \autoref{sec:des}, and assume for the
time being that $F\colon \Set\to\Set$ is a basic functor with only one
argument.

\subsection{Refinement Interfaces}\label{sec:refintface} New functors are
added to the framework by implementing a \emph{refinement interface}
(\autoref{D:refinement-interface}). The interface relates to an abstract
encoding of the functor and its coalgebras in terms of nodes and edges:
\begin{definition}\label{D:enc}
  An \emph{encoding} of a functor $F$ consists of a set $A$ of
  \emph{labels} and a family of maps
  $\flat\colon FX\to \Bag(A\times X)$, one for every set~$X$. The
  \emph{encoding} of an $F$-coalgebra $c\colon C\to FC$ is given by
  the map
  $\fpair{F!,\flat}\cdot c\colon C\to F1\times \Bag(A\times C)$ and we
  say that the coalgebra has $n = |C|$ states and
  $m = \sum_{x\in C}|\flat(c(x))|$ edges.
\end{definition}
\noindent
%
%
An encoding does by no means imply a reduction from $F$-coalgebras to
$\Bag(A\times (-))$-coalgebras, i.e.\ the notions of behavioural
equivalence for $\Bag(A\times (-))$ and~$F$, respectively, can be
radically different.  The encoding just fixes a representation format,
and $\flat$ is not assumed to be natural (in fact, it fails to be
natural in all encodings we have implemented except the one for
polynomial functors). Encodings typically match how one intuitively
draws coalgebras of various types as certain labelled graphs. For
instance for Markov chains (see Fig.~\ref{fig:example-input}),
i.e.~coalgebras for the distribution functor~$\Dist$, the set of
labels is the set of probabilities $A = [0,1]$, and
$\flat\colon \Dist X\to \Bag([0,1]\times X)$ assigns to each finite
probability distribution $\mu\colon X\to [0,1]$ the bag
$\{ (\mu(x),x) \mid x\in X, \mu(x) \neq 0\}$.

The implementation of a basic functor contains two ingredients:
(1)~a parser that transforms the syntactic specification of an
input coalgebra (\autoref{sec:parsing}) into the encoded
coalgebra, and (2)~the implementation of the refinement
interface.%
\smnote{Wouldn't it be better to move the definition of
  refinement interface here?}
\smnote{Some explanation that the user has to come up with $A$ and
  $\flat$ for a new functor is needed somewhere.}

\begin{wrapfigure}[10]{r}[0pt]{4cm}%
  \vspace*{-5pt}
\centering%
\raisebox{0pt}[\dimexpr\height-1.6\baselineskip\relax]{\begin{tikzpicture}
  \node (x) at (0.65,1.5) {$x$};
  \begin{scope}[
    every node/.append style={
      shape=circle,
      draw=none,
      fill=black,
      inner sep=0pt,
      outer sep=2pt,
      minimum width=1pt,
      minimum height=4pt,
      anchor=center,
    }
    ]
    \node (y1) at (0,0) {};
    \node (y2) at (1.3,0) {};
    \node (y3) at (2,0) {};
  \end{scope}
  \node[anchor=west] (y4) at (y3.east) {$\ldots$};
  \draw (y1) edge[draw=none] node {$\ldots$} (y2);
  \begin{scope}[
    every node/.append style={
      draw=black!30,
      rounded corners=2mm,
    }
    ]
    \node[fit=(y1) (y2)] (S) {};
    \node[fit=(y3) (y4)] (CminusS) {};
    \node[fit=(S) (CminusS)] (C) {};
  \end{scope}
  \begin{scope}[
    bend angle=10,
    space/.style={
      draw=white,
      line width=4pt,
    },
    edge/.style={
      draw=black,
      ->,
      every node/.append style={
        fill=white,
        shape=rectangle,
        inner sep=1pt,
        anchor=base,
        pos=0.55,
      },
    },
    ]
    \draw[space,bend right] (x) edge (y1);
    \draw[edge,bend right] (x) edge node[alias=s1] {$a_1$} (y1);
    \draw[space,bend left] (x) edge (y2);
    \draw[edge,bend left] (x) edge node[alias=sk] {$a_k$} (y2);
    \draw[space,bend left] (x) edge (y3);
    \draw[edge,bend left] (x) edge node[alias=c] {$b$} (y3);
    \coordinate (dotends) at ([xshift=10mm]c);
    \draw[edge,-,bend left] (x) edge (dotends);
    \draw[edge,-,shorten >=3mm,dotted] (dotends) edge +(4mm,-4mm);
  \end{scope}
  \draw (s1) edge[draw=none] node {$\cdots$} (sk);
  \draw (c) edge[draw=none] node {$\cdots$} (dotends);
  \begin{scope}[every node/.append style={anchor=north,yshift=-1mm,font=\footnotesize}]
    \draw[decorate,decoration={brace,mirror,amplitude=4pt}]
    ([yshift=-6mm]C.south west) -- node {$B$}
    ([yshift=-6mm]C.south east);
    \draw[decorate,decoration={brace,mirror,amplitude=4pt}]
    ([yshift=-2mm]S.south west) -- node {$S$}
    ([yshift=-2mm]S.south east);
    \draw[decorate,decoration={brace,mirror,amplitude=4pt}]
    ([yshift=-2mm]CminusS.south west) -- node {$B\setminus S$}
    ([yshift=-2mm]CminusS.south east);
  \end{scope}
\end{tikzpicture}}
\vspace{-1mm}
\caption{Splitting a block}
\noshowkeys
\label{fig:splitblock}
\end{wrapfigure}
To understand the motivation behind the
definition of a refinement interface, suppose that the generic
partition refinement has already computed some block of states
$B\subseteq C$ in its partition and that states in $S\subseteq B$
have different behaviour than those in $B\setminus S$. From this
information, the algorithm has to infer whether states $x,y\in C$ that
are in the same block and have successors in $B$ exhibit different
behaviour and thus have to be separated. For example, in the classical
Paige-Tarjan algorithm~\cite{PaigeTarjan87}, i.e.~for $F=\Pow$, $x$
and $y$ can stay in the same block provided that (a)~$x$ has a
successor in $S$ iff $y$ has one and (b)~$x$ has a successor in
$B\setminus S$ iff $y$ has one. Equivalently,
$\Pow\fpair{\chi_S,\chi_{B\setminus S}}(c(x)) =
\Pow\fpair{\chi_S,\chi_{B\setminus S}}(c(y))$,
where $\chi_S\colon C\to 2$ is the usual characteristic function of
the subset $S\subseteq C$. In the example of Markov chains,
i.e.~$F= \Dist$, $x,y\in C$ can stay in the same block if
$\sum_{x'\in S} c(x)(x') = \sum_{y' \in S} c(y)(y')$ and
$\sum_{x'\in B\setminus S} c(x)(x') =\sum_{y'\in B\setminus S}
c(y)(y')$,
i.e.~if
$\Dist\fpair{\chi_S,\chi_{B\setminus S}}(c(x)) =
\Dist\fpair{\chi_S,\chi_{B\setminus S}}(c(y))$.
Note that the element $(1,1)$ is not in the image of
$\fpair{\chi_S,\chi_{B\setminus S}}\colon C\to 2\times 2$. Since,
moreover, $S\subseteq B$, we can equivalently consider the map
\begin{equation}\label{eq:ax}
  \chi_S^B\colon C\to 3,
  \quad
  \chi_S^B(x \in S) = 2,
  \quad
  \chi_S^B(x \in B\setminus S) = 1,
  \quad
  \chi_S^B(x \in C\setminus B) = 0.
\end{equation}
That is, two states $x,y\in C$ can stay in the same block in the
refinement step provided that \(
F\chi_S^B(c(x)) = F\chi_S^B(c(y)).
\)
Thus, it is the task of a refinement interface to compute
$F\chi_S^B\cdot c$ efficiently and incrementally.
\smnote{TODO: read above and below and think about how to improve to
  address referee \#2.}
\begin{definition}\label{D:refinement-interface}
  Given an encoding $(A,\flat)$ of the set functor $F$,
  a \emph{refinement interface} for $F$ consists of 
  a set $W$ of \emph{weights} and functions
  \[
    \op{init}\colon F1×\Bag A\to W \qquad\text{and}\qquad
    \op{update}\colon \Bag A × W \to W× F3× W
  \]
  %
  %
  satisfying the following coherence condition: There exists a family of
  \emph{weight maps} $w\colon \mathcal{P} X \to (FX \to W)$ such that for all $t\in
  FX$ and all sets $S\subseteq B\subseteq X$,
  \[
    \begin{array}[b]{rl}
    w(X)(t) &= \op{init}\big(F!(t), \Bag \pi_1( \flat (t))\big)
    \\[1mm]
    (w(S)(t),F\chi_S^B(t),w(B\!\setminus\! S)(t))
    &= \op{update}\big(\{a\mid (a,x)\in \flat(t), x\in S\}, w(B)(t)\big).
    \end{array}
  \]
\end{definition}
Note that the comprehension in the first argument of $\op{update}$ is
to be read as a \emph{multiset} comprehension. In contrast to
$\op{init}$ and $\op{update}$, the function $w$ is not called by the
algorithm and thus does not form part of the refinement
interface. However, its existence ensures the correctness of our
algorithm. Intuitively, $X$ is the set of states of the input
coalgebra $(C,c)$, and for every $x \in C$, $w(B)(c(x)) \in W$ is the
overall weight of edges from $x$ to the block $B\subseteq C$ in
the coalgebra $(C,c)$. The axioms in
\autoref{D:refinement-interface} assert that $\op{init}$ receives in
its first argument the information which states of $C$ are
(non-)terminating, in its second argument the bag of labels of all
outgoing edges of a state $x \in C$ in the graph representation of
$(C,c)$, and it returns the total weight of those edges.  The
operation \op{update} receives a pair consisting of the bag of labels
of all edges from some state $x \in C$ into the set $S \subseteq C$
and the weight of all edges from $x$ to $B\subseteq C$, and from only
this information (in particular $\op{update}$ does not know $x$, $S$,
and $B$ explicitly) it computes the triple consisting of the weight
$w(S)(c(x))$ of edges from $x$ to $S$, the result of
$F\chi_S^B \cdot c(x)$ and the weight $w(B\setminus S)(c(x))$ of edges
from $x$ to $B\setminus S$ (e.g.~in the Paige-Tarjan algorithm, the
number of edges from $x$ to $S$, the value for the three way split,
and the number of edges from $x$ to $B\setminus S$,
cf.~\autoref{fig:splitblock}). Those two computed weights are needed
for the next refinement step, and $F\chi_S^B \cdot c(x)$ is used by
the algorithm to decide whether or not two states $x,y \in C$ that
are contained in the same block and have some successors in $B$ remain
in the same block for the next iteration.
%
%
%

For a given functor $F$, it is usually easy to derive the operations
$\op{init}$ and $\op{update}$ once an appropriate choice of the set
$W$ of weights and weight maps~$w$ is made, so we describe only the
latter in the following; see~\cite{concur2017,concurSpecialIssue} for
full definitions.

\begin{example}
  \begin{enumerate}
  \item For $F=\R^{(-)}$ we can take $W=\R^{2}$, and $w(B)(t)$ records
    the accumulated weight of $X\setminus B$ and $B$:
    $w(B)(t) = \big(\sum_{x\in X\setminus B} t(x), \sum_{x\in
      B}t(x)\big)$, i.e.~$w(B) = F\chi_B\colon FX\to F2$.
  \item More generally, let $F$ be one of the functors
    $G^{(-)}, \Bag,\Sigma$ where $G$ is a group and $\Sigma$ a
    signature with bounded arity, represented as a polynomial
    functor. Then we can take $W = F2$ (e.g.\ $W=\R^2$ for
    $F=\R^{(-)}$ as above) and $w(B) = F\chi_B\colon FX\to F2$.
  \item For $F=\Pow$, we need $W= 2 \times \N$, and
    $w(B)(t) = (|t\setminus B| \ge 1,|t\cap B|)$ records whether there
    is an edge to $X\setminus B$ and counts the numbers of edges
    into $B$.
  \end{enumerate}
\end{example}

\noindent In order to ensure that iteratively splitting blocks using
$F\chi_S^B$ in each iteration correctly computes the minimization of
the given coalgebra, we require that the type functor~$F$ is
\emph{zippable}, i.e.\ the evident maps
\( {\fpair{F(X+!),F(!+Y)}}\colon {F(X+Y)}\longrightarrow F(X+1)\times
F(1+Y) \) are injective~\cite[Definition~5.1]{concurSpecialIssue}. All
functors mentioned in \autoref{exManyFunctors} are zippable, and
zippable functors are closed under products, coproducts, and
subfunctors~\cite[Lemma~5.4]{concurSpecialIssue} but not under functor
composition; e.g.~$\Pow\Pow$ fails to be
zippable~\cite[Example~5.9]{concurSpecialIssue}.

%
%
The main correctness result~\cite{concurSpecialIssue} states that for
a zippable functor equipped with a refinement interface, our algorithm
correctly minimizes the given coalgebra. The low time complexity of our
algorithm hinges on the time complexity of the implementations of
$\op{init}$ and $\op{update}$. We have shown
previously~\cite[Theorem~6.22]{concurSpecialIssue} that if both
$\op{init}$ and $\op{update}$ run in linear time in the input list (of
type $\Bag A$) \emph{alone} (independently of $n,m$), then our generic
partition refinement algorithm runs in time $\CO((m + n)\cdot \log n)$
on coalgebras with $n$ states and $m$ edges. In order to cover
instances where the run time of $\op{init}$ and $\op{update}$ depends
also on~$n,m$, we now generalize this to the following new result:
\begin{theorem}\label{T:comp}
  Let $F$ be a zippable functor equipped with a refinement interface.
  Suppose further that $\rifactor(n,m)$ is a function such that in every run of
  the partition refinement algorithm on $F$-coalgebras with~$n$ states
  and~$m$ edges,
  \begin{enumerate}
  \item\label{T:comp:1} all calls to $\op{init}$ and $\op{update}$ on
    $\ell \in \Bag A$ run in time $\CO(|\ell|\cdot \rifactor(n,m))$;
  \item\label{T:comp:2} all comparisons of values of type $W$ run in
    time $\CO(\rifactor(n,m))$.
  \end{enumerate}
  Then the algorithm runs in overall time
  $\CO((m+n)\cdot \log n \cdot \rifactor(n,m))$.
\end{theorem}
Obviously, for $\rifactor(n,m)\in\CO(1)$, we obtain the
previous complexity. Indeed, for the functors $G^{(-)}$, $\Pow$,
$\Bag$, where $G$ is an abelian group, we can take $\rifactor(n,m) = 1$; this
follows from our previous work~\cite[Examples~6.4
and~6.6]{concurSpecialIssue}.  For a \emph{ranked alphabet}~$\Sigma$,
i.e.~a signature with arities of operations bounded by, say,~$r$, we
can take $\rifactor(m,n) = r \in \CO(1)$ if $\Sigma$ (or
just~$r$) is fixed.
%
%
We will discuss in \autoref{sec:wta} how \autoref{T:comp} instantiates
to weighted systems, i.e.~to monoid-valued functors
$M^{(-)}$ for unrestricted commutative monoids $M$.

\subsection{Combining Refinement Interfaces}\label{sec:des}%
In addition to supporting genericity via direct implementation of the
refinement interface for basic functors, our tool is \emph{modular} in
the sense that it automatically derives a refinement interface for
functors built from the basic ones according to the grammar
\eqref{termGrammar}. In other words, for such a functor the user does
not need to write a single line of new code. Moreover, when the user
implements a refinement interface for a new basic functor, this
automatically extends the effective grammar.

For example, our tool can minimize systems of type
$FX= \Dist(\N \times \Pow X\times \Bag X)$. To achieve this, a given
$F$-coalgebra is transformed into one for the functor
\(
  F' X = \Dist X +  (\N \times X \times X)  + \Pow X + \Bag X.
\)
\begin{figure}[t]
  \centering
  \begin{tikzpicture}
    \node[unary] (Potf) {\Dist};
    \node[binary] (times) at (Potf.in) {\Sigma};
    \node[unary] (Bagf) at ([yshift=2mm]times.in1) {\Pow};
    \node[unary] (Dist) at ([yshift=-2mm]times.in2) {\Bag};
    \path[draw=black!50]
      (times.out) edge node[above] {$Y$} (Potf.in)
      (Bagf.out) edge node[above] {$Z_1$} (times.in1)
      (Dist.out) edge node[below] {$Z_2$} (times.in2)
      (Dist.in) edge node[above] {$X$} +(8mm,0)
      (Potf.out) edge node[above] {$X$} +(-8mm,0)
      (Bagf.in) edge node[above] {$X$} +(8mm,0);
  \end{tikzpicture}
  \caption{Visualization of $FX = \Dist (\Sigma(\Pow X, \Bag X))$ for $\Sigma (Z_1,Z_2) = \N\times Z_1\times Z_1$}
  \label{fig:composedFunctor}
  \vspace{-5mm}
\end{figure}%
This functor is obtained as the sum of all basic functors involved in
$F$, i.e.~of all the nodes in the visualization of the functor term
$F$ (\autoref{fig:composedFunctor}). Then the components of the
refinement interfaces of the four involved functors $\Dist$, $\Sigma$,
$\Pow$, and $\Bag$ are combined by disjoint union $+$.  The
transformation of a coalgebra $c\colon C \to FC$ into a $F'$-coalgebra
introduces a set of intermediate states for each edge in the
visualization of the term $F$ in
\autoref{fig:composedFunctor}. E.g.~$Y$ contains an intermediate
state for every $\Dist$-edge,
i.e.~$Y = \{ (x,y) \mid \mu(x)(y) \neq 0\}$. Successors of such
intermediate states in~$Y$ lie in $\N \times Z_1 \times Z_2$, and
successors of intermediate states in $Z_1$ and $Z_2$ lie in $\Pow X$
and $\Bag X$, respectively. Overall, we obtain an $F'$-coalgebra on
$X+Y+Z_1+Z_2$, whose minimization yields the minimization of the
original $F$-coalgebra. The correctness of this construction is
established in full generality
in~\cite[Section~7]{concurSpecialIssue}.

\copar moreover implements a further optimization of this procedure
that leads to fewer intermediate states in the case of polynomial
functors~$\Sigma$: Instead of putting the refinement interface
of~$\Sigma$ side by side with those of its arguments, \copar includes
a systematic procedure to combine the refinement interfaces of the
arguments of~$\Sigma$ into a single refinement interface.  For
instance, starting from $FX=\Dist(\N \times \Pow X\times \Bag X)$ as
above, a given $F$-coalgebra is thus transformed into a coalgebra for
the functor $ F'' X = \Dist X + \N \times \Pow X \times \Bag X $,
effectively inducing intermediate states in~$Y$ as above but
avoiding~$Z_1$ and~$Z_2$. 

\subsection{Implementation Details}\label{sec:imp} Our
implementation is geared towards realizing both the level of
genericity and the efficiency afforded by the abstract
algorithm. Regarding genericity, each basic functor is defined (in its
own source file) as a single Haskell data type that implements two
type classes: a class that directly corresponds to the refinement
interface given in \autoref{D:refinement-interface} with its methods
\texttt{init} and \texttt{update}, and a parser that defines the
coalgebra syntax for the functor.  This means that new basic functors
can be implemented without modifying any of the existing code, except
for registering the new type in a list of existing
functors (refinement interfaces are
  in~\href{https://git8.cs.fau.de/software/copar/tree/master/src/Copar/Functors}{\texttt{src/Copar/Functors}}).

A key data structure for the efficient implementation of the generic
algorithm are refinable partitions, which store the current partition
of the set $C$ of states of the input coalgebra during the execution
of the algorithm.  This data structure has to provide constant time
operations for finding the size of a block, marking a state and
counting the marked states in a block. Splitting a block in marked and
unmarked states must only take linear time in the number of marked
states of this block.  In \copar, we use such a data structure
described (for use in Markov chain lumping) by Valmari and
Franceschinis~\cite{ValmariF10}.

Our abstract algorithm maintains two partitions~$P,Q$ of $C$, where
$P$ is one transition step finer than $Q$; i.e.~$P$ is the partition of $C$
induced by the map $Fq \cdot c\colon C \epito FQ$, where
$q\colon C \epito Q$ is the canonical quotient map assigning to every
state the block which contains it\lsnote{This is not very clear
  and, according to Thorsten, also not quite correct. @Thorsten:
  Please reword}. The key to the low time complexity is to choose in
each iteration a \emph{subblock}, i.e.~a block $S$ in $P$ whose
surrounding \emph{compound block}, i.e. the block~$B$ in $Q$ such that
$S \subseteq B$, satisfies $2 \cdot |S| \leq |B|$, and then refine $Q$
(and $P$) as explained in \autoref{sec:refintface} (see
\autoref{fig:splitblock})\lsnote{@Thorsten: Fix autoref to abbreviate
  Fig.}. This idea goes back to Hopcroft~\cite{Hopcroft71}, and is
also used in all other partition refinement algorithms mentioned in
the introduction.  Our implementation maintains a queue of subblocks
$S$ satisfying the above property\lsnote{@Thorsten: And only one
  partition, right? Please clarify}, and the termination condition
$P = Q$ of the main loop then translates to this queue being
empty. 

One optimization that is new in \copar in relation
to~\cite{ValmariF10,concurSpecialIssue} is that weights for blocks of
exactly one state are not computed, as those cannot be split any
further. This has drastic performance benefits for inputs where the
algorithm produces many single-element blocks early on, e.g.~for
nearly minimal systems or fine grained initial partitions,
see~\cite{Deifel18} for details and measurements.

\section{Instances}
\label{sec:instances}

Many systems are coalgebras for functors composed according to the
grammar~\eqref{termGrammar2}. In \autoref{tab:instances}, we list various system
types that can be handled by our algorithm, taken
from~\cite{concur2017,concurSpecialIssue} except for weighted tree automata, which are
new in the present paper. In all cases,~$m$ is the number of edges and $n$ is
the number of states of the input coalgebra, and we compare the run time of our
generic algorithm with that of specifically designed algorithms from the
literature. In most instances we match the complexity of the best known
algorithm. In the one case where our generic algorithm is asymptotically slower
(LTS with unbounded alphabet), this is due to assuming a potentially very large
number of alphabet letters -- as soon as the number of alphabet letters is
assumed to be polynomially bounded in the number~$n$ of states, the number~$m$
of transitions is also polynomially bounded in~$n$, so $\log m\in\CO(\log n)$.
This argument also explains why `$<$' and `$=$', respectively, hold in the last
two rows of \autoref{tab:instances}, as we assume~$\Sigma$ to be (fixed and)
finite; the case where~$\Sigma$ is infinite and unranked is more complicated.
Details on the instantiation to weighted tree automata are discussed in
\autoref{sec:wta}. We comment briefly on some further instances and initial
partitions:

\newcommand{\mycite}[1]{\cite{#1}}
\begin{table}[t] \centering
  \caption{Asymptotic complexity of the generic algorithm (2017/2019) compared to
    specific algorithms, for systems with $n$ states and $m$ transitions, respectively
    $m_\Pow$ nondeterministic and $m_\Dist$ probabilistic
    transitions for Segala systems. For simplicity, we assume that $m\ge
    n$ and, like~\cite{Hopcroft71,HoegbergEA07},  that $A$, $\Sigma$
    are finite.
    \smallskip}
    \label{tab:instances}
    \setlength\heavyrulewidth{0.25ex}%
\renewcommand{\arraystretch}{1.1}%
\renewcommand{\color}[1]{}%
\begin{tabular}{@{}l@{\hspace{4mm}}l@{\hspace{3mm}}c@{\hspace{4mm}}c@{\hspace{4mm}}crr@{}}
    \toprule
    System
    & Functor
    & Run-Time
    &
    & \multicolumn{3}{@{}l@{}}{Specific algorithm (Year)}
    \\
    \toprule
  \makecell[l]{
  DFA
  }
    & $2\times(-)^A$
    & $n\cdot \log n$
    & {$\bm=$}
    & $n\cdot \log n$
    & 1971
    & \mycite{Hopcroft71}
    \\
    \midrule
    \makecell[l]{Transition\\ Systems}
    & $\Pow$
    & $m\cdot \log n$
    & {$\bm=$}
    & \makecell{
      $m\cdot \log n$ \\
    }
    & 1987
    &  \mycite{PaigeTarjan87}
    \\
    \midrule
     \multirow{2}{*}{\makecell[l]{Labelled TS}}
    & \multirow{2}{*}{$\Pow(\N\times -)$}
    & \multirow{2}{*}{\makecell{$m\cdot \log m$}}
    & {$\bm=$}
    & $m\cdot \log m$
    & 2004
    & \mycite{DovierEA04}
    \\
    & 
    & 
    & {$\bm >$}
    & $m\cdot \log n$
    & 2009
    & \mycite{Valmari09}
    \\
    \midrule
    \makecell[l]{Markov\\Chains}
    & $\R^{(-)}$
    &$m\cdot \log n$
    & {$\bm=$}
    & $m\cdot \log n$
    & 2010
    & \mycite{ValmariF10}
    \\
    \midrule
    \multirow{2}{*}{\makecell[l]{Segala Systems}}
    & \multirow{2}{*}{\makecell[l]{$\Pow(A\times -)\cdot \mathcal{D}$}}
    & \multirow{2}{*}{\makecell{$m_\Dist \cdot \log m_\Pow$}}
    & {$\bm <$}
    & $m\cdot \log n$
    & 2000
    & \mycite{BaierEM00}
    \\
    & 
    & 
    & {$\bm =$}
    & \makecell{$m_\Dist \cdot \log m_\Pow$}
    & 2018
    & \cite{GrooteEA18}
    \\
    \midrule
    \makecell[l]{Colour\\Refinement}
    & $\Bag$
    & $m\cdot \log n$
    & $\bm=$
    & $m\cdot \log n$
    & 2017
    & \cite{BerkholzBG17}
    \\
  \midrule\midrule%
  \multirow{2}{*}{\makecell[l]{
  Weighted Tree\\[-1pt] Automata
  }}
    & $M\times M^{(\Sigma (-))}$
    & $m\cdot \log^2 m$
    & $\bm<$
    & $m\cdot n$
    & 2007
    & \cite{HoegbergEA07}
      \\
    & \makecell{$M\times M^{(\Sigma (-))}$\\[-1mm] \scriptsize ($M$ cancellative)}
    & \!\!$m\cdot \log m$\!\!
    & $\bm =$
    & \makecell{$m \cdot \log n$}
    & 2007
    & \cite{HoegbergEA07}
    \\
    \bottomrule
  \end{tabular}%
  \vspace*{-10pt}
\end{table}%

\subsubsection*{Further system types} can be handled by our algorithm and
tool by combining functors in various ways. For instance, general
Segala systems are coalgebras for the functor
$\Pow\Dist(A\times (-))$, and are minimized by our algorithm in time
$\CO((m+n)\cdot \log n)$, improving on the best previous
algorithm~\cite{BaierEM00}\twnote{I was not aware that
Baier e.a.\cite{BaierEM00} mention \emph{general} segala systems. Is this really
true?}; other type functors for various species of
probabilistic systems are listed in \cite{BARTELS200357}, including
the ones for reactive systems, generative systems, stratified systems,
alternating systems, bundle systems, and Pnueli-Zuck systems.

\subsubsection*{Initial partitions:} Note that in the classical
Paige-Tarjan algorithm \cite{PaigeTarjan87}, the input includes an initial
partition. Initial partitions as input parameters are covered via the
genericity of our algorithm: Initial partitions on $F$-coalgebras are
accomodated by moving to the functor $F'X = \N\times FX$, where the
first component of a coalgebra assigns to each state the number of its
block in the initial partition. Under the optimized treatment of the
polynomial functor \mbox{$\N\times(-)$} (\autoref{sec:des}), this
transformation does not enlarge the state space and also leaves the
complexity parameter $\rifactor(n,m)$ unchanged~\cite{concurSpecialIssue};
that is, the asymptotic run time of the algorithm remains unchanged
under adding initial partitions. 
%
%

%
%

\section{Weighted Tree Automata}
\label{sec:wta}
We proceed to take a closer look at weighted tree automata as a worked
example. In our previous work, we have treated the case where the
weight monoid is a group; in the present paper, we extend this
treatment to unrestricted monoids. As indicated previously, it is this
example that mainly motivates the refinement of the run time analysis
discussed in \autoref{sec:refintface}, and we will see that in the
case of non-cancellative monoids, the generic algorithm improves on
the run time of the best known specific algorithms in the literature.

Weighted tree automata simultaneously generalize tree automata and
weighted (word) automata. A partition refinement construction for
weighted automata (w.r.t.~weighted bisimilarity) was first considered
by Buchholz~\cite[Theorem~3.7]{Buchholz08}. Högberg et al.\ first
provided an efficient partition refinement algorithm for tree
automata~\cite{HoegbergEA09}, and subsequently for weighted tree
automata~\cite{HoegbergEA07}. Generally, tree automata differ from
word automata in replacing the input alphabet, which may be seen as
sets of unary operations, with an algebraic signature~$\Sigma$:

\begin{definition}\label{D:wta}
  Let $(M,+,0)$ be a commutative monoid. A (bottom-up) \emph{weighted tree
    automaton} (WTA) (over $M$) consists of a finite set $X$ of
  states, a finite signature $\Sigma$, an output map $f\colon X\to M$,
  and for each $k\ge 0$, a transition map
  $\mu_k\colon \Sigma_k\to M^{X^k\times X}$, where $\Sigma_k$ denotes
  the set of $k$-ary input symbols in $\Sigma$; the maximum arity of
  symbols in~$\Sigma$ is called the \emph{rank}.
\end{definition}
\twnote{example?}
\noindent
A weighted tree automaton is thus equivalently a finite coalgebra for
the functor $M\times M^{(\Sigma)}$ (where
$M^{(\Sigma)}(X) = M^{(\Sigma X)}$) where $\Sigma\colon \Set\to\Set$
is a polynomial functor. Indeed, we can regard the output map as a
transition map for a constant symbol, so it suffices to consider just
the functor $M^{(\Sigma)}$ (and in fact the output map is ignored in
the notion of backward bisimulation used by Högberg et
al.~\cite{HoegbergEA07}). For weighted systems, \emph{forward} and
\emph{backward} notions of bisimulation are considered in the
literature~\cite{Buchholz08,HoegbergEA07}; we do not repeat the
definitions here but focus on backward bisimulation, as it
corresponds to behavioural equivalence:
\begin{proposition}\label{wtaBackward} Backward bisimulation of
  weighted tree automata coincides with behavioural equivalence of
  $M^{(\Sigma)}$-coalgebras.
\end{proposition}

\takeout{
\begin{example}
  Weighted tree automata themselves subsume two other notions:
  \begin{enumerate}
  \item For $M$ being the Boolean monoid (w.r.t.~\emph{or}), we obtain
    non-determinstic tree automata and their bisimilarity.
  \item If all operations $\sigma \in \Sigma$ have arity one, then
    $F$-coalgebras are weighted automata.
  \end{enumerate}
\end{example}
}

\noindent Since $M^{(\Sigma)}$ is a composite of~$M^{(-)}$ and a
polynomial functor~$\Sigma$, the modularity of our approach implies
that it suffices to provide a refinement interface for~$M^{(-)}$.  For
the case where~$M$ is a group, a refinement interface with
$\rifactor(n,m) = 1$ has been given in our previous work. For the general
case, we distinguish, like Högberg et al.~\cite{HoegbergEA07},
between cancellative and non-cancellative monoids, because we obtain a better
complexity result for the former.

\subsection{Cancellative Monoids}
\label{sec:cancellative}
Recall that a commutative monoid $(M,+,0)$ is \emph{cancellative} if
$a+b = a+c$ implies $b = c$.  It is well-known that every cancellative
commutative monoid~$M$ embeds into an abelian group~$G$ via the
Grothendieck construction. Hence, we can convert $M^{(-)}$-coalgebras
into $G^{(-)}$-coalgebras and use the refinement interface for
$G^{(-)}$ from our previous work, obtaining
\begin{theorem}\label{T:cancel}
  On weighted tree automata with~$n$ states,~$k$ transitions and
  rank $r$ over a cancellative monoid, our algorithm runs in time
  \( \CO((rk + n)\cdot \log(k+n) \cdot r)\).
\end{theorem}
\noindent
Note that $rk$ may be replaced with the number $m$ of edges of the
corresponding coalgebra. Thus, for a fixed signature and $m \geq n$,
we obtain the bound in \autoref{tab:instances}.

\subsection{Non-cancellative Monoids}\label{sec:noncan}
The refinement interface for $G^{(-)}$ for a group $G$ (in which the cancellative
monoid $M$ in \autoref{sec:cancellative} is embedded) crucially makes
use of inverses for the fast computation of the weights returned by
$\op{update}$. For a non-cancellative monoid~$(M,+,0)$, we instead
need to maintain bags of monoid elements and consider subtraction of
bags. For the encoding of $M^{(-)}$, we take labels
$A=M_{\neq 0}=M\setminus\{0\}$, and
$\flat(f) = \{\,(f(x), x) \mid x\in X, f(x)\neq 0\,\}$ for
$f\in M^{(-)}$. The refinement interface for $M^{(-)}$ has weights
$W = M\times\Bagf(M_{\neq 0})$ and
\[
  w(B)(f) = \textstyle\big(\sum_{x\in X\setminus B}f(x),\;
  (m\mapsto\big|\{ x\in B\mid f(x) = m\}\big|) \big)\in M\times\Bagf(M_{\neq 0});
\]
that is, $w(B)(f)$ returns the total weight of $X\setminus B$
under~$f$ and the bag of non-zero elements of~$M$ occurring in
$f$.  The interface functions
$\op{init}\colon M^{(1)}×\Bag M_{\neq 0}\to W$,
$\op{update}\colon \Bag M_{\neq 0} × W \to W× M^{(3)}× W$ are
\begin{align*}
  \op{init}(f, \ell) &= (0, \ell)\\
  \op{update}(\ell, (r, c)) &=
  ((r + \Sigma(c-\ell), \ell),
  (r, \Sigma(c-\ell), \Sigma(\ell)),
  (r + \Sigma(\ell), c - \ell)),
\end{align*}
where for $a, b\in \Bagf Y$, the bag $a-b$ is defined by
$(a-b)(y) = \max(0, a(y) - b(y))$; $\Sigma\colon\Bagf M\to M$ is the
canonical summation map defined by
$\Sigma(b) = \sum_{m \in M} b(m) \cdot m$; and we denote elements of
$M^{(3)}$ as triples over~$M$.

We implement the bags $\Bagf(M_{\neq 0})$ used in
$W=M\times\Bagf(M_{\neq 0})$ as balanced search trees with keys
$M_{\neq 0}$ and values $\N$, following Adams~\cite{Adams93}. In
addition, we store in every node the value $\Sigma(b)$, where $b$ is
the bag encoded by the subtree rooted at that node. Hence, for every
bag $b$, the value $\Sigma(b)$ is immediately available at the root
node of the search tree encoding $b$. It is not difficult to see that
maintaining those values in the nodes only adds a constant overhead
into the operations of our data structure for bags and that the size
of the search trees is bounded by $\min(|M|,m)$. Thus, we obtain:
\begin{proposition}\label{P:monoids-fast}
  The above function $\op{update}(\ell, (r,c))$ can be computed in
  $\CO(|\ell|\cdot\log\min(|M|, m))$, where $m$ is the number of all
  edges of the input coalgebra.
\end{proposition}
\begin{corollary}\label{C:ncancel}
  On a weighted tree automaton with $n$ states, $k$ transitions, and
  rank~$r$ over an (unrestricted) monoid~$M$, our algorithm runs in
  time \(\CO\big((rk +n)\cdot \log(k+n) \cdot (\log k + r)\big)\),
  respectively \(\CO((rk + n) \cdot \log(k+n) \cdot r)\)
  if~$M$ is finite.
\end{corollary}
\noindent
More precisely, the analysis using Theorem~\ref{T:comp} shows that
$rk$ can be replaced with the number $m$ of edges of the input
coalgebra. Assuming $m \geq n$ we thus obtain the
bound given in \autoref{tab:instances}.
In addition to guaranteeing a good theoretical complexity, our tool
immediately yields an efficient implementation. For the case of
non-cancellative monoids, this is, to the best of our knowledge, the
only available implementation of partition refinement for weighted
tree automata.

\subsection{Evaluation and Benchmarking}\label{sec:bench}
We report on a number of benchmarks that illustrate the practical
scalability of our algorithm instantiated for weighted tree
automata. Previous studies on the practical performance of partition
refinement on large labelled transition
systems~\cite{valmari2010simple,Valmari09} show that memory rather
than run time seems to be the limiting factor. Since labelled
transition systems are a special case of weighted tree automata, we
expect to see similar phenomena. Hence, we evaluate the maximal
automata sizes that can be processed on a typical current computer
setup: We randomly generate weighted tree automata for various
signatures and monoids, looking for the maximal size of WTAs that can
be handled with 16 GB of RAM, and we measure the respective run times
of our tool, compiled with GHC version 8.4.4 on a Linux
  system and executed on an Intel® Core™ i5-6500 processor with
  3.20GHz clock rate.
\begin{table}[t] \centering
  \caption{Processing times (in seconds) $t_p$ for parsing and $t_a$ for partition refinement on maximal weighted tree automata with~$n$ 
    states  and $50\cdot n$ random transitions fitting into 16 GB of memory. File sizes range from 117 MB to
    141 MB, and numbers~$m$ of edges from 11 million to 17 million. }%
  \label{tab:maxfilesize}
  \twnote{TODO: check the numbers again}
  \smallskip
\setlength{\tabcolsep}{2.2pt}
  \begin{tabular}{@{}lclrrclrrclrrclrrclrr@{}}
    \toprule
    $\Sigma X\mathrlap{\ =}$
    &$\,$& \multicolumn{3}{@{}c@{}}{$4\!\times\!X$}
    && \multicolumn{3}{@{}c@{}}{$4\!\times\!X^2$}
    && \multicolumn{3}{@{}c@{}}{$4\!\times\!X^3$}
    && \multicolumn{3}{@{}c@{}}{$4\!\times\!X^4$}
    && \multicolumn{3}{@{}c@{}}{$4\!\times\!X^5$}
      \\
    \cmidrule{3-5}
    \cmidrule{7-9}
    \cmidrule{11-13}
    \cmidrule{15-17}
    \cmidrule{19-21}
    $M$
    && \multicolumn{1}{@{}c@{}}{$n$} & \multicolumn{1}{@{}c@{}}{$t_p$} & \multicolumn{1}{@{}c@{}}{$t_a$}
    && \multicolumn{1}{@{}c@{}}{$n$} & \multicolumn{1}{@{}c@{}}{$t_p$} & \multicolumn{1}{@{}c@{}}{$t_a$}
    && \multicolumn{1}{@{}c@{}}{$n$} & \multicolumn{1}{@{}c@{}}{$t_p$} & \multicolumn{1}{@{}c@{}}{$t_a$}
    && \multicolumn{1}{@{}c@{}}{$n$} & \multicolumn{1}{@{}c@{}}{$t_p$} & \multicolumn{1}{@{}c@{}}{$t_a$}
    && \multicolumn{1}{@{}c@{}}{$n$} & \multicolumn{1}{@{}c@{}}{$t_p$} & \multicolumn{1}{@{}c@{}}{$t_a$}
    \\
    \midrule
    $2$
    && 132177 & 53 & 188
    &&  98670 & 46 & 243
    &&  85016 & 47 & 187
    &&  59596 & 41 & 146
    &&  49375 & 38 & 114
    \\
    $\N$
&& 113957 & 61 & 141
&&  92434 & 55 & 175
&&  69623 & 49 & 152
&&  57319 & 47 & 140
&&  48962 & 45 & 112
    \\
    $2^{64}$
&& 114888 & 58 & 100
&&  95287 & 54 & 138
&&  70660 & 49 & 107
&&  62665 & 48 & 92
&&  49926 & 44 & 72
    \\
    \bottomrule
  \end{tabular}
  \vspace*{-15pt}
\end{table}
We fix $|\Sigma|=4$ and evaluate all combinations of rank~$r$ and
weight monoid~$M$ for $r$ ranging over $\{1,\dots, 5\}$ and $M$ over
$2=(2,\vee,0)$, $\N=(\N,\max,0)$ (the additive monoid of the tropical
semiring), and
$2^{64}=(2,\vee,0)^{64}\cong (\Powf(64),\cup,\emptyset)$. We write~$n$
for the number of states, $k$ for the number of transitions, and~$m$
for the number of edges in the graphical presentation; in fact, we
generate only transitions of the respective maximal rank~$r$, so
$m=k(r+1)$.  \autoref{tab:maxfilesize} lists the maximal values of~$n$
that fit into the mentioned 16~GB of RAM when $k= 50\cdot n$, and
associated run times.  For $M=(2,\vee,0)$, the optimized refinement
interface for $\Powf$ needs less memory, allowing for higher values
of~$n$, an effect that decreases with increasing rank~$r$. We restrict
to generating at most 50 different elements of $M$ in each automaton,
to avoid situations where all states are immediately distinguished in
the first refinement step. In addition, the parameters are
chosen so that with high likelihood, the final partition distinguishes
all states, so the examples illustrate the worst case. The first
refinement step produces in the order of $|\Sigma|\cdot \min(50, |M|)^r$
subblocks (cf.~\autoref{sec:imp}), implying earlier termination for
high values of~$|M|$ and~$r$ and explaining the slightly longer
run time for $M=(2,\vee,0)$ on small~$r$. We note in summary that WTAs with well
over $10$ million edges are processed in less than five minutes, and
in fact the run time of minimization is of the same order of magnitude
as that of input parsing. Additional evaluations on DFAs, Segala Systems, and
benchmarks for the Prism model checker~\cite{prismbenchmarks}, as well as a
comparison with existing specific tools by Antti Valmari~\cite{ValmariF10} and
from the mCRL2 toolset~\cite{BunteEA19} are in the full
version of this paper~\cite{DeifelEA19}.

\section{Conclusion and Future Work}\label{section:conclusion}
\label{sec:conc}

\takeout{
Future: 
\begin{itemize}
\item optimize refinement interfaces so that for functors that do not need
  weights, those don't need to be implemented and stored.
\item we are slower than Valmari for labelled transition systems
  because he can use bucket type of sorting (one bucket per label) during the splitting
  step. It remains to be seen whether and how this can be incorporated
  in our algorithm and implementation.
\item see CONCUR paper
\end{itemize}}

\noindent We have instantiated a generic efficient partition
refinement algorithm that we introduced in recent
work~\cite{concur2017,concurSpecialIssue} to weighted (tree) automata,
and we have refined the generic complexity analysis of the algorithm
to cover this case. Moreover, we have described an implementation of the
generic algorithm in the form of the tool \copar, which supports the
modular combination of basic system types without requiring any
additional implementation effort, and allows for easy incorporation of
new basic system types by implementing a generic refinement interface.

In future work, we will further broaden the range of system types that
our algorithm and tool can accomodate, and provide support for base
categories beyond sets, e.g.~nominal sets, which underlie nominal
automata~\cite{BojanczykEA14,SchroderEA17}.

\enlargethispage{2pt}
Concerning genericity there is an orthogonal approach by Ranzato and
Tapparo~\cite{RanzatoT08} that is generic over \emph{notions of
  process equivalence} but fixes the system type to standard labelled
transition systems; see also~\cite{DBLP:journals/tocl/GrooteJKW17}. Similarly, Blom and
Orzan~\cite{BlomOrzan03,BlomOrzan05} present \emph{signature
  refinement}, which covers, e.g.~strong and branching bisimulation
as well as Markov chain lumping, but requires adapting the algorithm
for each instance. These algorithms have also been improved using
symbolic techniques (e.g.~\cite{vDvdP18}). Moreover, many of the mentioned approaches and
others~\cite{BergaminiEA05,BlomOrzan03,BlomOrzan05,GaravelHermanns02,vDvdP18}
focus on parallelization. We will explore in future work whether
symbolic and distributed methods can be lifted to coalgebraic
generality. A further important aim is genericity also along the axis
of process equivalences.


\takeout{
\noindent We have presented a tool that efficiently minimizes systems
w.r.t.~coalgebraic behavioural equivalence. Both the algorithm and its
implementation are highly generic, and thus provide a tool that can be
used off-the-shelf as an efficient minimizer for many different types
of state-based systems, e.g.~deterministic automata, weighted tree
automata, ordinary, weighted, and probabilistic transition systems,
and Segala systems (which combine nondeterminism and probabilistic
choice). Our tool can easily be extended to accomodate new system
types, either by combining existing basic types by composition and
polynomial constructions, or by implementing a simple functor
interface for a new basic system type. Remarkably, there are instances
where the generic algorithm is asymptotically faster than the best
algorithms in the literature, notably on weighted systems over
non-cancellative weight monoids.

Future development of the generic algorithm will include further
broadening of the class of functors covered, in particular with a view
to neighbourhood systems, whose coalgebraic type functor fails to be
zippable. Another important extension is to include support for base
categories (of state spaces) beyond sets, with nominal sets (the base
category for nominal automata~\cite{BojanczykEA14,SchroderEA17}) as a
particularly promising candidate.

Several optimizations of the current implementation can be envisaged;
e.g.\ the implementation could be extended to detect when the set of
weights is trivial and simplify the refinement interface.
}

\bibliography{refs}
\bibliographystyle{myabbrv}
%
%
\clearpage
\appendix

\usetikzlibrary{shapes}
\tikzstyle{lambdatree}=[
    level distance = 5mm,
    every node/.style={
        inner sep=1pt,
        minimum width=4mm,
        minimum height=4mm,
        draw=black,
        shape=circle,
        outer sep=0pt,
    },
    noedge/.style={
        edge from parent/.style={}
    },
    child anchor=north,
    level 1/.style={sibling distance=0mm, level distance=5mm}, 
    level 2/.style={sibling distance=11mm, level distance=4mm},
    level 3/.style={sibling distance=9mm, level distance=6mm}, 
]

\tikzstyle{root}=[
draw=none,
minimum width=1mm,
minimum height=1mm,
]
\tikzstyle{subtree}=[
  isosceles triangle,
  draw,
  shape border rotate=90,
  isosceles triangle stretches=true,
  isosceles triangle apex angle=50,
  minimum height=8mm,
  minimum width=8mm,
  inner sep=2pt,
  yshift={0mm},
  anchor=north,
]
\tikzstyle{only math nodes}=[%
  every node/.append style={%
   execute at begin node=$,%
   execute at end node=$%
  }%
]
\chapter*{Appendix}
\makeatletter
\def\tableofcontents{\section*{\contentsname\@mkboth{{\contentsname}}%
                                                    {{\contentsname}}}
 \def\authcount##1{\setcounter{auco}{##1}\setcounter{@auth}{1}}
 \def\lastand{\ifnum\value{auco}=2\relax
                 \unskip{} \andname\
              \else
                 \unskip \lastandname\
              \fi}%
 \def\and{\stepcounter{@auth}\relax
          \ifnum\value{@auth}=\value{auco}%
             \lastand
          \else
             \unskip,
          \fi}%
 \@starttoc{toc}\if@restonecol\twocolumn\fi}
\makeatother
\tableofcontents
\section{Omitted Proofs}

\subsection{Proof of \autoref{T:comp}}
The case where $\rifactor(n,m) = 1$ is proved in
\cite[Theorem~6.22]{concurSpecialIssue}. We reduce the general case to
this one as follows. Observe that the previous complexity analysis
counts the number of basic operation performed by the algorithm
(e.g.~comparisons of values of type $W$) including those performed by
$\op{init}$ and $\op{update}$. In that analysis $\op{init}$ and
$\op{update}$ were assumed to have run time in $\CO(|\ell|)$, and
the total number of basic operations of the algorithm is then in
$\CO((m+n) \cdot \log n)$.

The number of steps taken by our algorithm under the current
assumptions~\ref{T:comp:1} and~\ref{T:comp:2} is thus bounded above by
the run time of the algorithm under the above assumptions
(i.e.~assuming $\rifactor(n,m) = 1$) but assuming that every basic operation
takes $\rifactor(n,m)$ steps. Hence, clearly the overall run time lies in
$\CO((m+n) \cdot \log n \cdot \rifactor(n,m))$. \qed


\subsection{Proof of  \autoref{wtaBackward}}
Given a weighted tree automaton $(X, f, (\mu_k)_{k \in \N})$ as in
\autoref{D:wta} we see that it is, equivalently, a finite coalgebra
for the functor $FX = M \times M^{(\Sigma X)}$, where we identify the
signature $\Sigma$ with its
corresponding polynomial functor $X\mapsto \coprod_{\sigma/k \in
  \Sigma} X^k$. Indeed $(\mu_{k})_{k\ge 0}$ is equivalently
expressed by a map
\begin{equation}
  \bar \mu \colon X\to M^{(\Sigma X)}
  \quad
  \text{with}
  \quad
  \bar \mu (x)(\sigma(x_1,\ldots,x_k))
  := \mu_k(\sigma)((x_1,\ldots,x_k), x).
  \label{eqMuBar}
\end{equation}
Thus we obtain a coalgebra
\begin{equation}\label{coalgWTADef}
  c\colon X\to M\times M^{(\Sigma X)}
  \quad
  \text{with}
  \quad
  c(x) = (f(x), \bar \mu(x)).
\end{equation}

Note that $\bar \mu(x)$ is finitely supported because $X$ and $\Sigma
X$ are finite.
\begin{definition}[Högberg et al.~{\cite[Definition~16]{HoegbergEA07}}]
  A \emph{backward bisimulation} on a weighted tree automaton
  $(X, f, (\mu_k)_{k \in \N})$ is an equivalence relation
  $R\subseteq X \times X$ such that for every $(p,q) \in R$, $\sigma/k\in \Sigma$, and
  $L\in \{D_1\times\cdots\times D_k\mid D_1,\ldots,D_k\in X/R\}$:
  \[
    \sum_{w\in L}
    \mu_k(\sigma)(w,p)
    =
    \sum_{w\in L}
    \mu_k(\sigma)(w,q).
  \]
\end{definition}
\begin{remark}
  Note that $w \in L$ means that $w\in X^k$ such that $e^k(w) = L$, where
  $e^k\colon X^k\epito (X/R)^k$ is the $k$-fold power of the canonical
  quotient map $e\colon X \to X/R$.
\end{remark}

\subsection{Details for \autoref{sec:cancellative}}
Every submonoid of a group is cancellative, for example $(\N,+,0)$ and
$(\Z,\cdot,1)$. Conversely via the \emph{Grothendieck construction},
every cancellative commutative monoid $(M,+,0)$ can be embedded into
the following group:
\[
  G = (M\times M)/\mathord{\equiv}
  \qquad
  (a_+,a_-)
  \equiv
  (b_+,b_-)
  ~\text{iff}~
  a_++b_- = a_++b_-,
\]
where the group structure is given by the usual componentwise
addition on the product:
$[(a_+,a_-)] + [(b_+,b_-)] = [(a_++b_+,a_-+b_-)]$, and
$-[(a_+,a_0)] = [(a_-,a_+)]$. The embedding of $M$ is given by
$m\mapsto [(m,0)]$.

The embedding $M\rightarrowtail G$ extends to a monic natural transformation
$M^{(-)}\rightarrowtail G^{(-)}$, and therefore, computing behavioural
equivalence for $M^{(-)}$ reduces to that of
$G^{(-)}$~\cite[Proposition~2.13]{concurSpecialIssue}.

\subsection{Proof of \autoref{wtaBackward}}
  We need to show that an
  equivalence relation $R\subseteq X \times X$ is a backward
  bisimulation iff the canonical quotient map
  $e\colon X\twoheadrightarrow X/R$ is an $M^{(\Sigma(-))}$-coalgebra
  homomorphism (with domain $(X,c)$ as defined in
  \eqref{coalgWTADef}). First, let $x\in X, \sigma/k\in \Sigma$, and
  $L\in (X/R)^k$ for some equivalence relation $R$. Then we have the
  following equalities, where note that 
  $M^{(\Sigma e)}\colon M^{(\Sigma X)}\to M^{(\Sigma(X/R))}$ and
  $\sigma(L) \in \Sigma(X/R)$:
  \begin{align*}
    \sum_{w\in L} \mu_k(\sigma)(w,x)
    &= \sum_{\mathclap{\substack{w\in X^k\\e^k(w) =L}}} \mu_k(\sigma)(w,x) =
    \sum_{\mathclap{\substack{\sigma(w)\in \Sigma X\\\op{ar}(\sigma) = k\\
          \Sigma e(\sigma(w)) =\sigma(L)}}} \mu_k(\sigma)(w,x) =
    \sum_{\mathclap{\substack{\sigma(w)\in \Sigma X\\
          \Sigma e(\sigma(w)) =\sigma(L)}}} \bar \mu (x)(\sigma(w))
    \\ & =
    \sum_{\mathclap{\substack{t\in \Sigma X\\
          \Sigma e(t) =\sigma(L)}}} \bar \mu (x)(t)
    = M^{(\Sigma e)}(\bar\mu(x))(\sigma(L)).
  \end{align*}
  Hence, for every equivalence relation $R\subseteq X\times X$ we have the
  following chain of equivalences:
  \begin{align*}
    &\text{$R$ is a backward bisimulation}
    \\
    \Leftrightarrow~
    & \forall (p,q)\in R, \sigma/k\in \Sigma, L\in (X/R)^k\colon
    \sum_{w\in L} \mu_k(\sigma)(w,p) =\sum_{w\in L} \mu_k(\sigma)(w,q)
    \\
    \Leftrightarrow~
    & \forall (p,q)\in R, \sigma/k\in \Sigma, L\in (X/R)^k\colon
    \\
    & \qquad\qquad
    M^{(\Sigma e)}(\bar\mu(p))(\sigma(L))
    = M^{(\Sigma e)}(\bar\mu(q))(\sigma(L))
    \\
    \Leftrightarrow~
    & \forall (p,q)\in R, t\in \Sigma(X/R)\colon
      M^{(\Sigma e)}(\bar\mu(p))(t)
      = M^{(\Sigma e)}(\bar\mu(q))(t)
      \\
    \Leftrightarrow~
    & \forall (p,q)\in R\colon
      M^{(\Sigma e)}(\bar\mu(p)) = M^{(\Sigma e)}(\bar\mu(q))
      \\
    \Leftrightarrow~
    & \forall (p,q)\in R\colon
      (M^{(\Sigma e)}\cdot\bar\mu)(p) = (M^{(\Sigma
        e)}\cdot\bar\mu)(q). 
    \end{align*}
    The last equation hold precisely if $e$ is a coalgebra
    homomorphism. Indeed, those equations holds precisely if there
    exists a map $r\colon X/R\to M^{(\Sigma(X/R))}$ such that $r(e(x))
    = (M^{(\Sigma e)}\cdot\bar\mu)(x)$, i.e.~such that the following
    square commutes
    \[
      \begin{tikzcd}[row sep = 7mm,baseline=(B.base)]
        X
        \arrow{r}{\bar \mu}
        \arrow[->>]{d}[swap]{e}
        &
        M^{(\Sigma X)}
        \arrow[->>]{d}{M^{(\Sigma e)}}
        \\
        X/R
        \arrow[dashed]{r}{r}
        &
        |[alias=B]| M^{(\Sigma (X/R))}
      \end{tikzcd}
      \tag*{\qed}
    \]

\takeout{
\subsection{Notes on Cancellative Monoids}\smnote[inline]{Do we still need this
  text?}

\smnote[inline]{This is old text, no?}
Recall that a commutative monoid $(M,+,0)$ is \emph{cancellative} iff $a+b =
a+c$ implies $b=c$. The Grothendieck group construction for cancellative monoids
defines the abelian group $G:=(M\times M)/\mathord{\sim}$, where $\sim$ is the equivalence
relation given by
\[
  (a,b) \sim (a',b') \quad\text{ iff }\quad
  a+b' = a'+b
\]
The transitivity of $\sim$ follows from the monoid being cancellative. The group
addition is given component-wise, the inverse of $[(a,b)]$ is $[(b,a)]$, and the
inclusion $h\colon M\hookrightarrow G$, $h(a) = [(a,0)]$ is a monoid
homorphism.
}

\subsection{Proof of \autoref{T:cancel}}
The functor $FX = M \times M^{(\Sigma(-))}$ is decomposed into $F''X =
M \times M^{(X)} + \Sigma X$ according to \autoref{sec:des}. Given a
coalgebra $\langle o,w\rangle\colon X \to M \times M^{(\Sigma X)}$ we
introduce the set of intermediate states $Y$, one for every outgoing
transition from every $x \in X$, i.e.
\[
  Y = \{(x,s) \mid \text{$x \in X$ and $s \in \Sigma X$ with $w(x,s)
    \neq 0$}\}.
\]
The given coalgebra structure $\langle o,w\rangle$ yields the two
evident maps $\xi_1\colon X \to M \times M^{(Y)}$ and $\xi_2\colon Y
\to \Sigma X$ given by
\begin{align*}
  \xi_1(x) &= (o(x), t_x), \quad 
  \text{where $t_x(x',s) = 
    \begin{cases}
      w(x,s) & \text{if $x' = x$} \\
      0 & \text{else,}
    \end{cases}$}
  \\
  \xi_2(x,s) &= s.
\end{align*}
Partition refinement is now performed on the following
$T''$-coalgebra:
\[
  X + Y \xrightarrow{\xi_1 + \xi_2} M \times M^{(Y)} + \Sigma X
  \xrightarrow{\id_M \times M^{(\inr)} + \Sigma\inl} M \times M^{(X+Y)} + \Sigma(X+Y),
\]
where $\inl\colon X \to X + Y$ and $\inr\colon Y \to X+Y$ denote the evident injections into the
disjoint union.

Clearly, we have $n' := |X + Y| = n + k$ and the number of edges of the
above coalgebra is at most $m' := (r+1) \cdot k$. Since the refinement interface for
$T''$ is a combination of those of $M\times M^{(-)}$ and $\Sigma$ its
factor $\rifactor(n',m')$ is the maximum of the factors $\rifactor_M(n',m')$ and
$\rifactor_\Sigma(n',m')$ of those two
refinement interfaces, respectively, since we either call the former or the latter
one for a state in $X+Y$; in symbols:
\[
  \rifactor(n',m') = \max(\rifactor_M(n',m'),\rifactor_\Sigma(n',m')) = \max (1,r),
\]
where we use that $M$ is cancellative and $\rifactor_\Sigma(n',m') = r$ (see
the end of \autoref{sec:refintface}).

By \autoref{T:comp} we thus obtain an overall time complexity of
\[
  \CO((m'+n') \cdot \log n' \cdot r)
  =
  \CO((r \cdot k + n)\cdot \log (n+k) \cdot r).
  \tag*{\qed}
\]

\subsection{Details for \autoref{sec:noncan}}

\begin{proposition}
  The refinement interface for $M^{(-)}$ is correct, for every monoid
  $(M,+,0)$.
\end{proposition}
\begin{proof}
  We prove that the refinement interface for $M^{(-)}$ defined in
  \autoref{sec:noncan} satisfies the two axioms in
  \autoref{D:refinement-interface}. Let $t\in FX$. For the first axiom
  we compute as follows:
  \begin{align*}
    w(X)(t) &= \big(\textstyle\sum_{x\in X\setminus X}t(x),\;
      (m\mapsto\big|\{ x\in X\mid t(x) = m\}\big|) \big)
              \\
    &= \big(0, \Bag\pi_1(\{(m,x)\in M_{\neq 0}\times X \mid t(x) = m\})\big)
      \\
    &= \big(0, \Bag\pi_1(\{(t(x),x)\mid x\in X, t(x) \neq 0\})\big)
    \\
    &= \big(0, \Bag\pi_1(\flat(t))\big)\\
    &= \op{init}\big(F!(t),\Bag\pi_1(\flat(t))\big).
  \end{align*}
  Let us now verify the second axiom concerning $\op{update}$. 
  In order to simplify the notation, we define the \emph{restriction} of bags of edges by
  \[
    (t\downarrow B) \in \BagM,
    \quad
    (t\downarrow B)(m) = |\{x\in B\mid t(x) = m\}|
    \quad
    \text{for }B\subseteq X, t\in M^{(X)}
  \]
  and define their sum by
  \newcommand{\xsum}[1]{\underset{#1}{{\textstyle\sum\,}}}
  \[
    \xsum{B} t := \csum(t\downarrow B) = \sum_{x\in B} t(x)\qquad\text{for
    }B\subseteq X.
  \]
  So we have
  \begin{align*}
    w(B)(t) &= \big(\textstyle\sum_{x\in X\setminus B}t(x), m\mapsto |\{x\in B\mid t(x) = m\}|\big)
            \\
      &= (\xsum{X\setminus B}t, (t\downarrow B))
  \end{align*}
  With the substraction of bags defined by 
  $(a-b)(y) = \max(0, a(y)-b(y))$ for $a,b\in \Bag Y$, we have
  \[
    (t\downarrow B) - (t\downarrow S)
    = (t\downarrow (B\setminus S))
    \qquad\text{for }S\subseteq B\subseteq X.
  \]
  Then for $S\subseteq B\subseteq X$ we can compute as follows:
  \begin{align*}
    &\phantom{=}\,\,\,\fpair{w(S),F\chi_S^B,w(B\!\setminus\! S)}(t)
      \\
    &= (w(S)(t),F\chi_S^B(t),w(B\!\setminus\! S)(t))
      \\
    &= \big((\xsum{X\setminus S}t,(t\downarrow S)),~
      (\xsum{X\setminus B} t, \xsum{B\setminus S}t, \xsum{S}t),~
      (\xsum{X\setminus (B\setminus S)}t,(t\downarrow (B\setminus S)))\big)
    \\
    &=
      \begin{array}[t]{rll}
      \big(
        & (\xsum{X\setminus B} t + \smash{\underbrace{\csum(t\downarrow (B \setminus S))}_{
          \xsum{B\setminus S}t}}
          , (t\downarrow S)),
      &(\xsum{X\setminus B} t, \csum(t\downarrow (B\setminus S)), \csum(t\downarrow S)),\\
      && (\xsum{X\setminus B} t + \xsum{S} t, ((t\downarrow B) - (t\downarrow S)))
      \big)
      \end{array}
    \\
    &=
      \begin{array}[t]{r@{}l}
      \big(
        & (\xsum{X\setminus B} t + \csum((t\downarrow B) - (t\downarrow S)), (t\downarrow S)), \\
      &(\xsum{X\setminus B} t, \csum((t\downarrow B) - (t\downarrow S)), \csum(t\downarrow S)),\\
      & (\xsum{X\setminus B} t + \csum(t\downarrow S), ((t\downarrow B) - (t\downarrow S)))
      \big)
      \end{array}
    \\
    &\overset{(*)}{=} \op{update}((t\downarrow S), (\xsum{X\setminus B} t, (t\downarrow B)))
    \\
    &= \op{update}\big((m\in M_{\neq 0}\mapsto |\{x\in S\mid t(x) = m\}|), w(B)(t)\big)
      \\
    &= \op{update}\big(\{m\in M_{\neq 0}\mid (m,x)\in \flat(t), x\in  S\}, w(B)(t)\big)\\
    &= \op{update}\big(\{a \in A\mid (a,x)\in \flat(t), x\in S\}, w(B)(t)\big),
  \end{align*}
  where the step labelled~$(*)$ uses the definition of $\op{update}$:
  \[
    \op{update}(\ell, (r, c)) =
    ((r + \Sigma(c-\ell), \ell),
    (r, \Sigma(c-\ell), \Sigma(\ell)),
    (r + \Sigma(\ell), c - \ell)).
    \tag*{\qed}
  \]
\end{proof}

\subsection{Proof of \autoref{P:monoids-fast}}
  We implement the bags $\Bagf(M_{\neq 0})$ used in
  $W=M\times\Bagf(M_{\neq 0})$ as balanced search trees with keys
  $M_{\neq 0}$ and values $\N$ following Adams~\cite{Adams93}. In
  addition, we store in every node $x$ the value $\Sigma(b)$, where $b$ is
  the bag encoded by the subtree rooted at $x$; in the following we
  simply write $\Sigma(x)$. Hence, for every
  bag $b$, the value $\Sigma(b)$ is immediately available at the root
  node of the search tree for $b$.

  Note that our search trees cannot have more nodes than the size $|M|$ of
  their index set, and the number of nodes is also not greater than
  the number $m$ of all edges. Hence, their size is bounded by
  $\min(|M|, m)$.

  Recall e.g.~\cite[Section~14]{CormenEA}, that the key operations $\op{insert}$,
  $\op{delete}$ and $\op{search}$ have logarithmic time
  complexity in the size of a given balanced binary search tree. 

  We still need to argue that maintaining the values $\Sigma(x)$ in
  the nodes does not increase this complexity. This is obvious for the
  $\op{search}$ operation as it does not change its argument search
  tree. For $\op{insert}$ and $\op{delete}$ recall from
  \emph{op.~cit.} that those operation essentially trace down one path
  starting at the root to a node (at most a leaf) of the given search
  tree. In addition, after insertion or deletion of a node, the
  balanced structure of the search tree has to be fixed. This is done by
  tracing back the same path to the root and (possibly) performing
  \emph{rotations} on the nodes occurring on that path. Rotations are
  local operations changing the structure of a search tree but
  preserving the inorder key ordering of subtrees (see
  \autoref{fig:rotations}).
  \begin{figure}
    \centering
    \begin{tikzpicture}
      \begin{scope}[lambdatree, only math nodes]
        \node[root] (t1) {}
        child { node{y}
          child { node {x}
            child { node[subtree] {\alpha} }
            child { node[subtree] {\beta} }
          }
          child { node[subtree] (gamma1) {\gamma} }
        } ;
      \end{scope}
      \begin{scope}[lambdatree, only math nodes]
        \node[root] (t2) at (5cm,0) {}
        child { node {x}
          child { node[subtree] (alpha2) {\alpha} }
          child { node{y}
            child { node[subtree] {\beta} }
            child { node[subtree] {\gamma} }
          }
        } ;
      \end{scope}
      \path[draw,->,shorten <= 2mm,shorten >= 2mm]
      (gamma1) edge[transform canvas={yshift=3mm}] node[above]{right rotation} (alpha2)
      (alpha2) edge node[below]{left rotation} (gamma1)
      ;
  \end{tikzpicture}
    \caption{Rotation operations in binary search trees.}\label{fig:rotations}
  \end{figure}
  
  Clearly, in order to maintain the correct summation values in a
  search tree under a rotation, we only need to adjust those values in
  the nodes $x$ and $y$. For example, for right rotation the new
  values are:
  \[
    \Sigma(y) = \Sigma(\beta) + \Sigma(\gamma);\qquad\Sigma(x) =
    \Sigma(\alpha) + \Sigma(y).
  \]
  In addition, when inserting or deleting a node $x$ we must recompute
  the $\Sigma(y)$ of all nodes $y$ along the path from the root to $x$
  when we trace that path back to the root. This can clearly be
  performed in constant time for each node $y$ since $\Sigma(y)$ is
  the sum (in $M$) of those values stored at the child nodes of $y$.

  In summary, we see that maintaining the desired summation values only
  requires an additional constant overhead in the backtracing
  step. Consequently, the operations of our balanced binary search
  trees have a run time in $\CO(\log \min(|M|,m))$, so that we obtain
  the desired overall time complexity of $\op{update}$.\qed

\subsection{Proof of \autoref{C:ncancel}}
  We proceed precisely as in the proof of \autoref{T:cancel} using
  that $\rifactor_M(n',m') = \log \min(|M|, m')$ (in lieu of $\rifactor_M(n',m') =
  1$). We proceed by case distinction. If $M$ is a finite
  monoid, then $\rifactor_M(n',m')$ is in $\CO(1)$. Hence, we obtain the same
  overall complexity as in \autoref{T:cancel} as desired. 

  If $M$ is infinite, we have that $\rifactor_M(n',m')$ is in $\CO(\log m')$. Thus, by
  \autoref{T:comp} we obtain:
  \begin{align*}
    \CO((m'+n') \cdot \log n' \cdot \max (\log m', r)) 
    & = \CO((rk +n) \cdot \log (n+k) \cdot (\log(rk) + r)) \\
    & = \CO((rk +n) \cdot \log (n+k) \cdot (\log k + r)),
  \end{align*}
  where we use that $\log(rk) = \log r + \log k$ in the second
  equation and, in the first one, that $\max$ can be replaced by $+$
  in $\CO$-notation.
  \qed
\begin{remark} Here we provide a more refined comparison of the
  complexity of the Högberg et al.'s algorithm with
  the instances of our algorithm for weighted tree automata.
  \begin{enumerate}
  \item For arbitrary (non-cancellative) monoids they provide a
    complexity of
    $\CO(r \cdot k \cdot n)$~\cite[Theorem~27]{HoegbergEA07}. Again,
    the number $m$ of edges of the input coalgebra satisfies
    $m \leq rk$. Moreover, for a fixed input signature, $m$ and $k$ are
    asymptotically equivalent. Assuming further that $m \geq n$, which
    means that there are no isolated states, we see that the bound in
    \autoref{tab:instances} indeed improves the complexity of
    $\CO(m \cdot n)$ Högberg et al.'s algorithm. To see this note
    first that the number $m$ of edges is in $\CO(n^{r+1})$ so that we
    obtain
    \[
      \CO(m\cdot \log(m)^2)
      \subsetneq
      \CO(m\cdot \sqrt[r+1]{m})
      \subseteq
      \CO(m\cdot \sqrt[r+1]{n^{r+1}})
      = \CO(m\cdot n).
    \]
    In the first step, it is used that for every $d \ge 1$ and $0 < c < 1$ we
    have $\CO(\log^d(m))\subsetneq \CO(m^{c})$.
  \item For cancellative monoids, the time bound given in
    \emph{op.~cit.}~is $\CO(r^2\cdot t\cdot\log n)$
    \cite[Theorem~29]{HoegbergEA07}. Under our standard assumption
    that $rt \geq n$ our complexity from \autoref{T:cancel} lies in
    $\CO(r^2 \cdot t \cdot \log(t+n)$, which is only very slightly
    worse.
  \end{enumerate}
\end{remark}

\section{Further Benchmarks}\label{S:addbench}
\subsection{More details for main benchmark}
The benchmarks are designed in such a way that the run time is maximal. This
means that:
\begin{enumerate}
\item In the first partition that is computed, all states are identified.
\item In the final partition, all states are distinguished.
\end{enumerate}

We minimize randomly generated coalgebras for
\[
  FX = M\times M^{(\Sigma X)}
\]
with $\Sigma X= 4\times X^r$, for $r\in \{1,\ldots,5\}$. When generating a
coalgebra with $n$ states, we randomly create 50 outgoing transitions per state,
leading to $50\cdot n$ transitions in total. As described in \autoref{sec:des},
we need to introduce one intermediate state per transition, leading to an
actual number $n' = 51\cdot n$ of states. Every transition of rank $r$ has one incoming
edge and $r$ outgoing edges, hence $m = (r+1)\cdot k = 50\cdot (r+1) \cdot n$.
The performance evaluation is listed in \autoref{tab:extended}.

\newcommand{\boolmon}{$(2,\vee,0)$}
\newcommand{\Zmax}{$(\N,\max,0)$}
\newcommand{\WordOr}{$(\Powf(64),\cup,\emptyset)$}
\begin{table}
  \caption{Extended version of \autoref{tab:maxfilesize}: We take $n$ states and
    50 transitions per state, leading to $n'$ states and $m$ edges in total. Parsing takes $t_p$ seconds;
    the first partition has $P_1$ blocks and its computation takes $t_i$ seconds;
    the final partition has $P_f$ blocks and its computation takes additional
    $t_r$ seconds.
    }
    \label{tab:extended}
    \medskip
  \centering
\begin{tabular}{@{}l@{\hspace{2mm}}lr@{~~}r@{~~}r@{~~}r@{~~}r@{~~}r@{~~}r@{~~}r@{~~}r@{}}
\toprule
r & Monoid $M$  & \multicolumn{1}{c}{$n$} & \multicolumn{1}{c}{$n'$} &
\multicolumn{1}{c}{$m$} & \multicolumn{1}{c}{Size} & \multicolumn{1}{c}{$P_1$} & \multicolumn{1}{c}{$P_f$} & \multicolumn{1}{c}{$t_p$} & \multicolumn{1}{c}{$t_i$} & \multicolumn{1}{c}{$t_r$} \\
\midrule
1 & \boolmon & 132177 &        6741027 & 13217700 & 117 MB    &         6 &     132177 & 53 & 32 & 156 \\
1 & \Zmax    & 114888 &        5859288 & 11488800 & 122 MB    &       416 &     114888 & 58 & 34 & 66 \\
1 & \WordOr  & 113957 &        5811807 & 11395700 & 131 MB    &        54 &     113957 & 61 & 32 & 109 \\
\midrule
2 & \boolmon &  98670 &        5032170 & 14800500 & 123 MB    &         6 &      98670 & 46 & 31 & 212 \\
2 & \Zmax    &  95287 &        4859637 & 14293050 & 136 MB    &       404 &      95287 & 54 & 30 & 108 \\
2 & \WordOr  &  92434 &        4714134 & 13865100 & 141 MB    &        54 &      92434 & 55 & 31 & 144 \\
\midrule
3 & \boolmon &  85016 &        4335816 & 17003200 & 138 MB    &         6 &      85016 & 47 & 20 & 167 \\
3 & \Zmax    &  70660 &        3603660 & 14132000 & 127 MB    &       397 &      70660 & 49 & 25 & 82 \\
3 & \WordOr  &  69623 &        3550773 & 13924600 & 132 MB    &        54 &      69623 & 49 & 25 & 127 \\
\midrule
4 & \boolmon &  59596 &        3039396 & 14899000 & 119 MB    &         6 &      59596 & 41 & 25 & 121 \\
4 & \Zmax    &  62665 &        3195915 & 15666250 & 136 MB    &       397 &      62665 & 48 & 26 & 66 \\
4 & \WordOr  &  57319 &        2923269 & 14329750 & 130 MB    &        54 &      57319 & 47 & 25 & 115 \\
\midrule
5 & \boolmon &  49375 &        2518125 & 14812500 & 116 MB    &         6 &      49375 & 38 & 24 & 90 \\
5 & \Zmax    &  49926 &        2546226 & 14977800 & 127 MB    &       376 &      49926 & 44 & 20 & 52 \\
5 & \WordOr  &  48962 &        2497062 & 14688600 & 129 MB    &        54 &      48962 & 45 & 20 & 92 \\
  \bottomrule
\end{tabular}
\end{table}

\subsection{Maximally dense WTAs}
As another benchmark, we generated very dense WTAs, showing that also in this
degenerated case, $m$ = 10 million edges (in the sense of graphical representation) can be
handled with 16 GB of RAM within a few minutes, as described in the following.

\begin{table}\centering
  \caption{
    Dense WTAs:
    For $(2,\vee,0)$, the final partition consists of very few blocks (less than
    20, depending on the polynomial). For the other cases, all states are
    distinguished after the initialization phase.
    }
    \label{tab:degenerated}
    \medskip
  \begin{tabular}{@{}l@{\hspace{2mm}}lrrrrrrrrr@{}}
    \toprule
    $M$ & $\Sigma X$ & $p$ & $n$ & $n' / 10^6$& $m/ 10^6$ & MiB & $t_p$ & $t_i$ & $t_r$ & $t$ \\
    \midrule
    $(2,\vee,0)$                & $8\times X$                 & 0.7 & 1478 & $5.24$ & $10.48$ & 82  & 35   & 28   & 37   & 115   \\
                                & $1+4\times X^2$             & 0.7 & 151  & $4.13$ & $12.39$ & 83  & $33$ & $21$ & $34$ & $102$ \\
                                & $4+3\times X + 2\times X^2$ & 0.7 & 190  & $4.15$ & $12.41$ & 83  & $33$ & $19$ & $73$ & $143$ \\
                                & $3\times X^5$               & 0.7 & 11   & $1.59$ & $9.57$  & 47  & $20$ & $11$ & $11$ & $45$  \\
    \midrule
    $(\Z,\max,0)$               & $8\times X$                 & 0.7 & 1450 & $5.05$ & $10.09$ & 182 & $43$ & $27$ & $40$ & $127$ \\
                                & $1+4\times X^2$             & 0.7 & 150  & $4.05$ & $12.15$ & 162 & $41$ & $22$ & $40$ & $121$ \\
                                & $4+3\times X + 2\times X^2$ & 0.7 & 188  & $4.02$ & $12.02$ & 162 & $40$ & $15$ & $39$ & $111$ \\
                                & $3\times X^5$               & 0.7 & 11   & $1.59$ & $9.57$  & 79  & $23$ & $11$ & $21$ & $60$  \\
    \midrule
    $(\Pow(64),\cup,\emptyset)$ & $8\times X$                 & 0.7 & 1408 & $4.76$ & $9.52$  & 164 & $49$ & $25$ & $37$ & $121$ \\
                                & $1+4\times X^2$             & 0.7 & 148  & $3.89$ & $11.67$ & 151 & $44$ & $21$ & $37$ & $118$ \\
                                & $4+3\times X + 2\times X^2$ & 0.7 & 186  & $3.89$ & $11.64$ & 152 & $44$ & $18$ & $38$ & $115$ \\
                                & $3\times X^5$               & 0.7 & 11   & $1.60$ & $9.57$  & 77  & $25$ & $11$ & $21$ & $61$  \\
    \bottomrule
  \end{tabular}
\end{table}

For a fixed
number $n = |C|$ of states, monoid $M$ and signature $\Sigma$, we
uniformly generated $c(x)\in M^{(\Sigma C)}$ for every state $x\in
C$:\footnote{We overload notation and write $\Sigma$ both for a
  signature and its associated polynomial functors on sets.}
for every $y\in \Sigma C$, we put $c(x)(y) = 0$ with probability 0.7
and otherwise choose $c(x)(y) \in M$ randomly. We then counted the
number $m$ of edges, checked whether the minimization stayed within
the RAM limit, and if so, recorded the run time.
\autoref{tab:degenerated} lists the sizes of the largest WTAs that can
be handled. For each configuration we generated five automata, and
averaged their values for \autoref{tab:degenerated}; each of the
maximal deviations is insignificantly small. Note that with higher rank $r$,
even a small number $n$ of states leads to millions of edges.

We see that the upper limit for the number of edges is roughly 10
million, independent of the choices of the monoid or the
signature. In more detail, only the choice of the
signature contributes to $n$ and nuances of $m$. This is hardly
surprising because the signature determines the branching degree of an
automaton and also determines that there are $0.3 \cdot |\Sigma C|$
many expected transitions. With roughly 10 million edges, the file
sizes vary slightly, depending on the representation of the
coefficients from the different monoids (coefficients vanish for
$(2,\vee,0)$).

In \autoref{tab:degenerated} we show the benchmarks of very dense WTAs.

\bigskip
\subsection{Benchmarks for DFAs and PRISM Models}
Another benchmark suite consists of randomly generated DFAs for a
fixed size $n$ and alphabet $A$, which were generated by uniformly
choosing a successor for each state and letter of the alphabet and for
each state whether it is final or not. All of the resulting automata
were already minimal, which means that our algorithm has to refine the
initial partition $X/Q = \{X\}$ until all blocks are singletons, which
requires a maximal number of iterations. The results in \autoref{tab:bench-dfa-1} and
\ref{tab:bench-dfa-2}
show that the implementation can handle coalgebras with 10 million
edges, and that parsing the input takes more time than the actual
refinement for these particular systems. \twnote{Which functor is
  used?  $2\times(-)^A$ or $2\times \Pow(A\times(-))$? use booktabs in
  the figures and add time for the initial partition. does state count
  mean before or after the desorting? HP: $n$ and $|A|$ are defined
  above and the functor is the first one (it's introduced as DFA, not
  LTS)}

\begin{table}[h]
  \hspace*{-2mm}%
  \rlap{%
  \begin{minipage}{\textwidth+6mm}%
    \hfill
  \begin{minipage}[t]{.27\textwidth}%
    \strut\vspace*{-\baselineskip}\newline%
  \begin{tabular}{@{}c@{\hspace{.5em}}r@{\hspace{0.7em}}cr@{}}
    \toprule
    \clap{~Input} & \multicolumn{3}{c}{Time (in s) to} \\
    $n$ & Parse
        & Init
        & Refine \\
    \midrule
    1000 & 2.40  & 0.76 & 0.36 \\
    2000 & 4.96  & 1.58 & 0.74 \\
    3000 & 7.39  & 2.11 & 1.40 \\
    4000 & 10.20 & 3.20 & 1.67 \\
    5000 & 13.06 & 4.05 & 2.10 \\
    \bottomrule
  \end{tabular}
  \end{minipage}%
  \hfill%
  \begin{minipage}[t]{.26\textwidth}%
    \strut\vspace*{-\baselineskip}\newline%
  \begin{tabular}{@{}r@{\hspace{.5em}}r@{\hspace{.5em}}rr@{}}
    \toprule
    \clap{Input~~~~} & \multicolumn{3}{c}{Time (in s) to} \\
    \multicolumn{1}{c}{$n~$}
                            & Parse
                            & Init
                            & Refine \\
    \midrule
    600 & 44.75 & 1.82 &  2.88 \\
    700 & 50.93 & 4.29 & 3.18 \\
    800 & 60.78 & 2.54 & 4.16 \\
    900 & 68.34 & 2.76 & 4.60 \\
    1000 & 75.79 & 3.05 & 5.21 \\
    \bottomrule
  \end{tabular}
  \end{minipage}%
  \hfill\hspace*{0mm}
  \\%
  \hspace*{0mm}\hfill
  \begin{subfigure}[t]{.27\textwidth}%
    \caption{DFAs for $|A|=10^3$}
    \noshowkeys\label{tab:bench-dfa-1}
  \end{subfigure}%
  \hfill%
  \begin{subfigure}[t]{.26\textwidth}%
    \caption{DFAs for $|A|=10^4$}
    \noshowkeys\label{tab:bench-dfa-2}
  \end{subfigure}%
  \hfill\hspace*{0mm}
  \end{minipage}%
  }%
  \\
  \caption{Performance on randomly generated DFAs}
    \label{tab:bench}
\end{table}
\newcommand{\benchparam}[1]{\scriptsize \textnormal{(}#1\textnormal{)}}
\begin{table} \centering
  \begin{tabular}{@{}lrr@{\hspace{.31em}}r@{\hspace{0.5em}}rrrr@{}}
    \toprule
    PRISM Model & \multicolumn{2}{c}{Input} & \multicolumn{3}{c}{Time (s) to}
                & \multicolumn{2}{c}{Time (s) of} \\
                        & \multicolumn{1}{c}{States}
                        & \multicolumn{1}{c}{Edges}
                        & Parse
                        & Init
                        & Refine
                        & Valmari
                        & mCRL2 \\
    \midrule
    \tt fms \benchparam{n=4} &   35910 &  237120 & 0.48 & 0.12 & 0.16 & 0.21 & \multicolumn{1}{c}{--} \\
    \tt fms \benchparam{n=5} &  152712 & 1111482 & 2.46 & 0.68 & 1.10 & 1.21 & \multicolumn{1}{c}{--} \\
    \tt fms \benchparam{n=6} &  537768 & 4205670 & 9.94 & 2.91 & 5.56 & 5.84 & \multicolumn{1}{c}{--} \\
    \tt \makecell[l]{wlan2\_collide\\\benchparam{COL=2,TRANS\_TIME\_MAX=10}}      &   65718 &   94452 &  0.51 & 0.29 &  0.59 & 0.14 & 0.42 \\
    \tt \makecell[l]{wlan0\_time\_bounded\\{\scriptsize \textnormal{(}TRANS\_TIME\_MAX=10,DEADLINE=100\textnormal{)}}}&  582327 &  771088 &  5.26 & 3.07 &  5.52 & 0.92 & 3.18 \\
    \tt \makecell[l]{wlan1\_time\_bounded\\\benchparam{TRANS\_TIME\_MAX=10,DEADLINE=100}} & 1408676 & 1963522 & 13.42 & 6.17 & 16.13 & 2.52 & 8.58 \\
    \bottomrule
  \end{tabular}
  \\[2mm]
  \caption{Performance on PRISM benchmarks}
  \noshowkeys\label{tab:bench-prism}
\end{table}

In addition, we translated the benchmark suite of the model checker
PRISM~\cite{KNP11} to coalgebras for the functors $\N\times \R^{(X)}$ for continuous
time markov chains (CTMC) and $\N\times \Pow(\N\times(\Dist X))$ for Markov decision
processes (MDP). In contrast to DFAs, these functors are compositions of several
basic functors and thus require the construction described in \autoref{sec:des}. Two
of those benchmarks~\cite{prismbenchmarks} are shown in \autoref{tab:bench-prism}
with different parameters, resulting in three differently sized coalgebras each.%
\footnote{The full set of benchmarks and their results can be found at
  \url{https://git8.cs.fau.de/software/copar-benchmarks}} The \emph{fms} family of
systems model a flexible manufacturing system as CTMC, while the \emph{wlan}
benchmarks model various aspects of the IEEE 802.11 Wireless LAN protocol as
MDP.\hpnote{The linked pages have citations for ``Flexible Manufacturing System''
  and ``WLAN''. Should we cite the same sources?}

\autoref{tab:bench-prism} also includes the total run time of two
additional partition refinement tools: A C++
implementation\footnote{Available at
  \url{https://git8.cs.fau.de/hpd/mdpmin-valmari}} of the algorithm
described in \cite{ValmariF10} by Antti Valmari which can minimize
MDPs as well as CTMCs, and the tool \texttt{ltspbisim} from the mCRL2
toolset~\cite{BunteEA19} version 201808.0 which implements a recently
discovered refinement algorithm for MDPs~\cite{GrooteEA18} (but does
not support CTMCs, hence there is no data in the first three lines).

The results in \autoref{tab:bench-prism} show that refinement for the \emph{fms}
benchmarks is faster than for the respective \emph{wlan} ones, even though the
first group has more edges. This is due to~(a) the fact that the functor
for MDPs is more complex and thus introduces more indirection into our
algorithms, as explained in \autoref{sec:des}, and (b)~that our optimization for
one-element blocks fires much more often for
\emph{fms}.\hpnote{\textbf{A lot} more often. That's the only reason we can
  compete with Valmari here. His algorithm even has better O-notation complexity
  in |A| (which we to consider to be fixed in this example, but still...)}

It is also apparent that \copar{} is slower than both of the other tools in our
comparison, by a factor of up to 15 for the presented examples. This performance
difference can be partly attributed to the fact that our implementation is
written in Haskell and the other tools are written in C++. In addition,
\copar{}'s genericity and modularity take a toll on performance.


\end{document}